\documentclass{article}

\usepackage[english]{babel}
\usepackage{amsthm,amsmath,lipsum}
\usepackage[algo2e]{algorithm2e} 
\usepackage{xcolor}
\usepackage{fullpage}
\usepackage{amssymb}
\usepackage{centernot}
\usepackage{tikz}
\usepackage{soul}
\usetikzlibrary{shapes.geometric, arrows}
\usetikzlibrary{positioning}
\usetikzlibrary{calc}

\DeclareFontFamily{U}{mathb}{\hyphenchar\font45}
\DeclareFontShape{U}{mathb}{m}{n}{
      <5> <6> <7> <8> <9> <10> gen * mathb
      <10.95> mathb10 <12> <14.4> <17.28> <20.74> <24.88> mathb12
}{}
\DeclareSymbolFont{mathb}{U}{mathb}{m}{n}
\DeclareMathSymbol{\llcurly}{3}{mathb}{"CE}
\DeclareMathSymbol{\ggcurly}{3}{mathb}{"CF}

\usepackage{thm-restate}

\theoremstyle{plain}

\newtheorem{theorem}{Theorem}
\newtheorem{lemma}{Lemma}
\newtheorem{definition}{Definition}
\newtheorem{observation}{Observation}
\newtheorem{corollary}{Corollary}
\newtheorem{claim}{Claim}

\usepackage{times}
\usepackage{soul}
\usepackage{url}
\usepackage[hidelinks]{hyperref}
\usepackage[utf8]{inputenc}
\usepackage[small]{caption}
\usepackage{graphicx}
\usepackage{amsmath}
\usepackage{amsthm}
\usepackage{amssymb}
\usepackage{booktabs}
\usepackage{booktabs}
\usepackage{algorithm}
\usepackage{algorithmic}
\usepackage{color}
\usepackage[switch]{lineno}

\usepackage[most]{tcolorbox}

\title{Maximin Share Guarantees for Few Agents with Subadditive Valuations}

\author{George Christodoulou\thanks{Aristotle University of
    Thessaloniki, and Archimedes/Athena RC, Greece. Email: \texttt{\{gichristo,vgchrist,smastra\}@csd.auth.gr} }
\and{Vasilis Christoforidis\footnotemark[1]}
\and {Symeon Mastrakoulis\footnotemark[1]}
\and Alkmini Sgouritsa\thanks{Athens University of Economics and Business, and Archimedes/Athena RC, Greece. Email: \texttt{alkmini@aueb.gr}}}

\date{}

\begin{document}

\maketitle

\begin{abstract}
    We study the problem of fairly allocating a set of indivisible items among a set of agents. We consider the notion of (approximate) maximin share (MMS) and we provide an improved lower bound of $1/2$ (which is tight) for the case of subadditive valuations when the number of agents is at most four. We also provide a tight lower bound for the case of multiple agents, when they are equipped with one of two possible types of valuations. Moreover, we propose a new model that extends previously studied models in the area of fair division, which will hopefully give rise to further research. We demonstrate the usefulness of this model by employing it as a technical tool to derive our main result, and we provide a thorough analysis for this model for the case of three agents. Finally, we provide an improved impossibility result for the case of three submodular agents. 
    
\end{abstract}
\usetikzlibrary {patterns,patterns.meta}
\section{Introduction}
Fair allocation of indivisible goods is a fundamental problem at the
intersection of economics and computer science. The goal is to allocate
a set of $m$ goods among a set of $n$ agents with heterogeneous
preferences. In contrast to the case of divisible goods (known as
cake-cutting) fundamental concepts of fairness—such as envy-freeness and proportionality— do not carry over when pivoting from the
continuous to the discrete domain. 

In this work, we investigate (approximate) Maximin Share fairness (MMS), a
well-studied concept introduced by
\cite{BudishMMS}. This notion can be seen as the discrete analogue of
proportionality in the context of indivisible goods.  The objective is to
allocate a bundle to each agent, ensuring that her value is at
least her {\em maximin share} (or a significant fraction of
it). Intuitively, an agent's maximin share is the maximum value she
could guarantee by proposing a partition of the items into
$n$ parts and then receiving the least desirable bundle from that
division. 

Although for two agents with additive valuations, MMS allocations always exist, \cite{KurokawaProcacciaWang18} showed that this is not always the case for three or more agents.
Studying approximate MMS fairness has led to a surge of research in the past
decade, for additive (e.g., \cite{AmanatidisMNS17,GhodsiHSSY21,GargTaki21,AkramiGarg24}) but also for more general valuation classes (e.g., \cite{BarmanKrishnamurthy20,GhodsiHSSY22,SeddighinSeddighin24}). 
The case of subadditive valuations is wide-open with a huge (polylogarithmic) gap between the best known upper and lower bounds. Closing this gap is considered a major open problem in this area.

 \subsection{Our contribution}

We study the existence of approximate MMS allocations, a predominant notion of fairness in settings with indivisible goods, under submodular
and subadditive valuations. We focus on settings with few agents. We are able to show the following results.

\begin{itemize}
    \item In our main technical result (Theorem~\ref{thm:mainTheorem}), we show the existence of $1/2$-MMS allocations for at most four agents with subadditive valuations. We note that this result improves the best known bounds for many other valuation classes, including OXS, Gross Substitutes, submodular and XOS. We emphasize that the guarantee of our theorem matches the best known impossibility result due to \cite{GhodsiHSSY22}, thus settling the case of four agents for both subbaditive and XOS valuations (see also Table~\ref{table:Results}). 

    \item We show the existence of $1/2$-MMS allocations for multiple agents when they have one of two admissible valuation functions (Theorem~\ref{thm:two-types}).

    \item On the way to prove Theorem~\ref{thm:mainTheorem} we develop a new model that is useful for inductive arguments. This model incorporates nicely previously studied MMS variants, and we believe that it is of independent interest. In Section~\ref{sec:Reductions}, we provide a thorough study of approximate guarantees for three agents and we are able to provide two complete characterizations: one for the case of $1/2$-MMS$(\mathbf{d})$ (Theorem~\ref{thm:three-general}) and one for the case of $(1,1/2,1/2)$-MMS$(\bf d)$ (Theorem~\ref{thm:three-general-approximate}). 

    \item We show an improved impossibility result for three agents with submodular valuation functions: we present an instance in which an $\alpha$-MMS allocation does not exist for $\alpha > 2/3$ (Theorem~\ref{thm:SubmodUB}). This improves the previous best known result of 3/4 due to \cite{GhodsiHSSY22}.

\end{itemize}

\paragraph{$\boldsymbol{\alpha}$-MMS(${\bf d}$).} In order to show Theorem~\ref{thm:mainTheorem}, we find it useful to study a notion that we call $\boldsymbol{\alpha}$-MMS(${\bf d}$), where $\boldsymbol{\alpha}=(\alpha_1,\ldots, \alpha_n)$ is a vector of threshold values $a_i\in [0,1]$, and  ${\bf d}=(d_1,\ldots, d_n)$ is a vector of positive integers. Roughly an allocation is $\boldsymbol{\alpha}$-MMS(${\bf d}$) if it guarantees to each agent $i$ an $\alpha_i$ approximation of the minimum value she could guarantee if she could partition all the items in $d_i$ parts  (see Definition~\ref{def:alpha-mms} for a formal definition).  This notion is very useful for showing approximate MMS guarantees using inductive arguments, as we demonstrate in the technical part. It effectively captures several previous studied notions, and we believe it is of independent interest. 
For example, $\mathbf{1}$-MMS($\mathbf{d})$ corresponds to ordinal approximations of the maximin share, i.e. the maximum value an agent $i$ can guarantee if she partitions the set of items into $d$ bundles \cite{BudishMMS}. 
Another example is the notion of $(\alpha, \beta)$-MMS introduced by \cite{HosseiniS21}, which guarantees to an $\alpha$ fraction of the agents a $\beta$ approximation of their MMS. This can be captured by our notation $\boldsymbol{\alpha'}$-MMS(${\bf d}$)  where $\boldsymbol{\alpha'}=(\beta, \ldots, \beta, 0,\ldots, 0)$ is a vector where $\alpha n$ of its elements are equal to $\beta$ and the rest are equal to $0$, and for $\mathbf{d} = (n, \dots, n)$. They further showed that the existence of $(\alpha, 1)$-MMS implies the existence of 1-out-of-$k$-MMS for $k \ge \lceil \frac{n}{\alpha} \rceil$. %thus establishing a connection between the $(\alpha, \beta)$ framework and ordinal approximations. 
Perhaps more closely related to our work, is the model of \cite{AkramiGST24} for approximate MMS guarantees with agent priorities, which captures both ordinal and multiplicative approximations as special cases, i.e. the concept of $\boldsymbol{\alpha}$-MMS({\bf d}), (for $\mathbf{d} = (n, \dots, n)$).

 \subsection{Further Related Work}

We refer the reader to the recent survey of \cite{AmanatidisABFLMVW23Survey} covering a wide variety of discrete fair division settings along with the main fairness notions and their properties. We focus on the concept of maximin share (MMS).

\paragraph{Maximin share and $\alpha$-MMS.} MMS has seen significant progress during the past years. Nevertheless, the case beyond additive utilities remains underexplored; we summarize the state-of-the-art bounds in Table~\ref{table:Results} for the most well-known superclasses of additive valuations. Additionally,  
 $1/2$-MMS allocations are known to exist for SPLC valuations \cite{ChekuriKKM24}. \cite{Hummel:HSS24} proved the same guarantee for the case of hereditary set systems, improving upon the $11/30$ guarantee of \cite{LiVetta21}. In sharp contrast, a surge of works has led to strong approximations for additive valuations \cite{KurokawaProcacciaWang18,AmanatidisMNS17,GhodsiHSSY21,GargTaki21,AkramiGST23}. Notably, \cite{AkramiGarg24} recently showed an approximation factor of $3/4+3/3836$.

\begin{table}[]
    \centering

    \begin{tabular}{c c c}
        \toprule
        Submodular & Existence & Non-existence
        \\
        \midrule
        $n=2$ & $2/3^{\dagger}$ & $2/3^{\ddagger}$ \\
        $n=3$ & $10/27^{\dagger\dagger}$ & 
        \begin{tabular}{c}
             $3/4^{\mathparagraph}$, \\
             $\mathbf{\color{red}{2/3}}$\textcolor{red}{[Thm \ref{thm:SubmodUB}]}
        \end{tabular}
        \\
        
        %$n \ge 4$ & $10/27^{\dagger\dagger}$ & $3/4^{\mathparagraph}$ \\
        \midrule
        XOS \\
        \midrule
        $n \le 4$ &
            \begin{tabular}{c}
              $3/13^{\mathsection}$, \\
             $\mathbf{\color{red}{1/2}}$ \textcolor{red}{[Thm \ref{thm:mainTheorem}]}
        \end{tabular}

        & $1/2^{\mathparagraph}$ \\
        
        %$n \ge 5$ & $3/13^{\mathsection}$ & $1/2^{\mathparagraph}$\\
        \midrule
        Subadditive \\
        \midrule
        $n \le 4$ &
                \begin{tabular}{c}
             $\Omega(\frac{1}{\log n \log \log n})^{**}$, \\
$\mathbf{\color{red}{1/2}}$ \textcolor{red}{[Thm \ref{thm:mainTheorem}]}
        \end{tabular}
            & $1/2^{\mathparagraph}$ \\ 
        %$n \ge 5$ & $\Omega(\frac{1}{\log n \log \log n})^{**}$ & $1/2^{\mathparagraph}$\\
        \bottomrule
    \end{tabular}
    \caption{Best known MMS approximations. Our contributions appear in \textbf{bold red} color. $^{\dagger}$ \cite{ChristodoulouChristoforidis}; $^{\ddagger}$\cite{KKM23}; $^{\dagger\dagger}$\cite{UziahuFeigeMMS}; $^{\mathparagraph}$\cite{GhodsiHSSY22}; $^{\mathsection}$\cite{AkramiMehlhornSeddighinShahkarami23}; $^{**}$\cite{SeddighinSeddighin24}.}
    
    \label{table:Results}
\end{table}

\paragraph{1-out-of-$d$-MMS.} \cite{BudishMMS} introduced the notion and showed the existence of 1-out-of-$(n+1)$-MMS, albeit with excess goods. \cite{Aigner-HorevSH22} obtained a bound of $d=2n-2$ on the existence of 1-out-of-$d$ allocations under additive valuations, which was subsequently improved to 1-out-of-$\lceil 3n/2 \rceil$ by \cite{HosseiniS21}, 1-out-of-$\lfloor 3n/2 \rfloor$ by \cite{HosseiniSSH22}, and recently to 1-out-of-$4\lceil n/3 \rceil$ by \cite{AkramiGST24}.
\cite{BabaioffNT21} introduced the $\ell$-out-of-$d$ maximin share, which corresponds to the maximum value an agent can guarantee to herself if she partitions the items into $d$ bundles and  then being allocated the worst $\ell$ of them.

\paragraph{Few Agents.} Settings with few agents have been previously studied in the fair division literature, particularly for MMS approximations involving three or four agents with additive valuations. The algorithm of \cite{KurokawaProcacciaWang18} provides a $3/4$-MMS guarantee for three or four agents. \cite{AmanatidisMNS17} subsequently established a bound of $\alpha \ge 7/8$ for the case of three agents, which was later improved to $8/9$ by \cite{GourvesMonnot19} and finally to $11/12$ by \cite{FeigeNorkinMMS}. \cite{GhodsiHSSY21} improved upon the previously known factor by showing a $4/5$ lower bound for four agents.

 \section{Model}
\label{sec:Model}

We consider a setting of allocating a set $M = \{1, \dots, m\}$ of indivisible items among a set $N = \{1, \dots, n\}$ of agents. 
Each agent $i \in N$ is equipped with a valuation function $v_i: 2^M \to \mathbb{R}_{\ge 0}$, and we denote by  $\boldsymbol{v} = (v_1, \dots, v_n)$ the valuation profile of all agents. 
We consider valuations that are monotone (i.e. $v(S) \le v(T)$, for all $S \subseteq T$) and normalized (v($\emptyset) = 0$). For brevity, we sometimes use $v_i(g)$ instead of $v_i(\{g\})$. 
In this work, we consider valuation classes that belong to the complement-free hierarchy, such as additive, submodular, XOS, and subadditive valuation functions which we define below.

\noindent{\bf Additive valuations.} A valuation function $v$ is additive if $v(S) = \sum_{g \in S} v(g)$ for any $S \subseteq M$.    

\noindent{\bf Submodular valuations.} A valuation function $v$ is submodular if $v(S) + v(T) \ge v(S \cup T) + v(S \cap T)$ for any $S, T \subseteq M$.

\noindent{\bf Fractionally subadditive (XOS) valuations.} A valuation function $v$ is XOS if there exists a set of additive functions $a_1, \dots, a_k$ such that $v(S) = \max_{l \in [k]}a_l(S)$.    

\noindent{\bf Subadditive valuations.} A valuation function $v$ is subadditive if $v(S) + v(T) \ge v(S \cup T)$ for any $S, T \subseteq M$.

It is well-known that $Additive \subsetneq Submodular \subsetneq XOS \subsetneq Subadditive$.

We are interested in allocating $M$ into $n$ mutually disjoint sets $A_1,\ldots, A_n$, where $A_i$ is the bundle of items assigned to agent $i$. We denote the respective allocation by $A=(A_1, \ldots, A_n)$. 

{\bf Minimum values:} We often present to an agent a partition $P=(S_1,\ldots, S_d)$ of $M$ into $d$ parts and we ask them how they value each part. We denote by $\Pi_d(M)$ the set of all possible such partitions. We define the minimum value of agent $i$ with respect to some fixed partition $P$  as
$$\mu_i^{P}(M) = \min_{1\leq j \leq d} v_i(S_j)\,.$$
We denote by $\mu_i^{d}(M)$ the {\em maximum} minimum guarantee (MMS($d$)) that agent $i$ can achieve by the best partition with $d$ parts; i.e.,
$$\mu_i^{d}(M) = \max_{P \in \Pi_{d}(M)}\mu_i^{P}(M).$$
 We refer to the $d$ bundles $S_1,\ldots, S_d$ that comprise the best partition $P$, as the MMS$(d)$ bundles of agent $i$ or just the MMS bundles of $i$ when $d$ is clear from the context.  For $d=n$, $\mu_i^{n}(M)$ is the {\em MMS value} of agent $i$, and we drop the superscript and simply denote it by $\mu_i(M)$.  Also, when $M$ and $d$ are clear from the context, we use the simpler notation $\mu_i^{d}$ or $\mu_i$. For simplicity we assume (by scaling) that all valuations are normalized such that $\mu_i^{d}=1$. 

 Given a vector of $n$ positive integers ${\bf d}=(d_1,\ldots, d_n)$ we define a vector of partitions ${\bf P:=P(d)}=(P_1,\ldots, P_n)$ with respect to ${\bf d}$,  where the
 partition $P_i\in \Pi_{d_i}(M)$ corresponds to the partition of agent
 $i$ into $d_i$ parts. We are interested in providing
 approximation guarantees $\alpha_i$ for the minimum value
 $\mu_i^{P_i}(M)$ of each agent $i$ with respect to ${\bf P}$. This is summarised in the following definition.
\begin{definition} [$\boldsymbol{\alpha}$-MMS$(\mathbf{P})$, $\boldsymbol{\alpha}$-MMS$(\mathbf{d})$]\label{def:alpha-mms}
Fix $\boldsymbol{\alpha}=(\alpha_1,\ldots,\alpha_n)$ with $\alpha_i \in [0, 1]$, and ${\bf d}=(d_1,\ldots, d_n)$ with $d_i\in \mathbb{N}_+$ for all $i\in N$. Fix also $\mathbf{P}=(P_1,\ldots, P_n)$, a fixed vector of partitions with $P_i\in \Pi_{d_i}(M)$ for each $i$. An allocation $A$ is $\boldsymbol{\alpha}$-MMS$(\mathbf{P})$ if $v_i(A_i) \ge \alpha_i \cdot \mu_i^{P_i}(M)$ for all $i \in N$.    
An allocation $A$ is  $\boldsymbol{\alpha}$-MMS$(\mathbf{d})$ if it is $\boldsymbol{\alpha}$-MMS$(\mathbf{P(d)})$ for all partition vectors ${\bf P(d)}$ w.r.t to $\bf{d}$.\footnote{Equivalently, an allocation $A$ is  $\boldsymbol{\alpha}$-MMS$(\mathbf{d})$ if $v_i(A_i) \ge \alpha_i \cdot \mu_i^{d_i}(M)$ for all $i \in N$.}
\end{definition}

Note that for $\boldsymbol{\alpha}=(\alpha, \ldots,\alpha)$ and $\mathbf{d}=(n, \dots, n)$, the definition of $\boldsymbol{\alpha}$-MMS$(\mathbf{d})$ coincides with the standard definition of $\alpha$-MMS and we drop the dependency on $\mathbf{d}$. 
Additionally, for simplicity when $\boldsymbol{\alpha}$ is uniform, i.e., $\alpha_i=\alpha$ for each agent $i$, we write $\alpha$-MMS$(\mathbf{d})$ instead of $\boldsymbol{\alpha}$-MMS$(\mathbf{d})$. We use the notation ${\bf d}_{-i}$ and $\boldsymbol{\alpha}_{-i}$ to refer to those vectors where their $i$-th element is omitted.

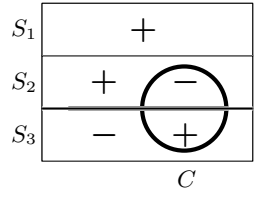
\begin{figure}
\begin{center}
\small{
    \begin{tikzpicture}[scale=0.7]
    \draw[step=1cm,] (0,-1) grid (4,2);
        \draw[left color = white, right color = white] (0,1) rectangle (4,2);
        \draw[left color = white, right color = white] (0,0) rectangle (4,1);
    \draw[left color = white, right color = white] (0,-1) rectangle (4,0);
    \draw[black,ultra thick](0.5,0) -- (3.5,0) arc(0:360:0.8) --cycle;3,0);
    \draw[white, ultra thick] (0,0) -- (4,0);
    \draw[black, thick] (0,0) -- (4,0);

\node[anchor=west] at (-0.75,1.5) {$S_1$};
\node[anchor=west] at (-0.75,0.5) {$S_2$};
\node[anchor=west] at (-0.75,-0.5) {$S_3$};
%\node[anchor=south] at (3.75,-0.75) {$C$};
\node[anchor=south] at (2.75,-1.65) {$C$};
\node[anchor=west] at (1.5,1.5) {\large{$\boldsymbol{+}$}};

\node[anchor=west] at (0.75,0.5) {\large{$\boldsymbol{+}$}};
\node[anchor=west] at (2.3,0.5) {\large{$\boldsymbol{-}$}};

\node[anchor=west] at (0.75,-0.5) {\large{$\boldsymbol{-}$}};
\node[anchor=west] at (2.3,-0.5) {\large{$\boldsymbol{+}$}};

\end{tikzpicture}}
\end{center}

    \caption{Let  $P=(S_1,S_2,S_3)$ be a partition for agent $S$. The maximum desired half over set $C$ is depicted by thick circle. For each bundle $S_i$, plus bundles (+) attain value at least  $v_S(S_i)/2$, while minus (-) attain value at most $v_S(S_i)/2$. Hence  $\mathcal{X}_{S}(C,P)=\{S_3\cap C\}$, $\mathcal{X}_{S}({M\setminus C},P)=\{S_1\setminus C,S_2 \setminus C \}$ and $\mathcal{X}^*_{S}(C,P)=\mathcal{X}_{S}({M\setminus C},P)$.}
\label{fig:maximumCut}
\end{figure}
Our allocation protocols proceed by repeatedly asking agents to
evaluate various subsets of items using specific types of valuation
queries. In the most common type of query, which we call a {\em cut}, we
present a subset $C\subseteq M$ to an agent $S$, ask her how she values
the intersection of $C$ with each of her MMS bundles $S_j$. Due to
subadditivity, at least one of $C\cap S_j$ or $C\setminus S_j$ will provide her
with satisfactory value, that is higher than $v_S(S_j)/2$. We call the side (intersection or complement) that has the highest number of satisfactory values a {\em Maximum
  Desired Half} (see definition below and Figure~\ref{fig:maximumCut} for an illustration).

\begin{definition} [Maximum Desired Half] Let $C \subseteq M$ and $P=\left(S_1,\ldots,S_r\right)\in \Pi_r(M)$ be a partition into $r$ bundles for an agent $S$. The set $\mathcal{X}_S(C,P)$ collects intersections of $C$ with each $S_i$ that have sufficiently high value (greater or equal to $1/2$) i.e.,
$$\mathcal{X}_S(C,P)=\left\{S_i \cap C : v_S\left(S_i\cap C\right)\ge \frac{v_S\left(S_i\right)}{2}, S_i \in P\right\}.$$
We define the Maximum Desired Half of agent $S$ over the set $C$ w.r.t partition $P$ as follows
$$\mathcal{X}_{S}^*(C,P) = \arg \max\left\{\lvert \mathcal{X}_S(C,P)\rvert,\lvert \mathcal{X}_S(M\setminus C,P) \rvert\right\}$$
We refer to the set $C$ as a cut $C$. When the partition $P$ is clear from the context (e.g., when the partition is the MMS bundles of $S$) we will drop the dependency on $P$.
\end{definition}

Intuitively, every bundle in the Maximum Desired Half of agent $S$
guarantees at least half the value of the minimum bundle in $P$ for
agent $S$ and moreover it is disjoint with either the cut $C$ or with
its complement w.r.t to each bundle $S_j$. The key property of the
Maximum Desired Half set is that it contains at least half of the
bundles.

\begin{observation}
\label{obs:Cut_noOfDesiredSets}
For every $C \subseteq M$ and $P=(S_1,\ldots,S_r)$, $\lvert \mathcal{X}^*_{S}(C,P)\rvert \ge \lceil r/2\rceil$.
\end{observation}
\begin{proof}
 Due to subadditivity, for each $S_i$ it holds that $v_S(S_i \cap C)+v_S(S_i \setminus C) \ge v_S(S_i)$, and therefore at least one of the two terms on the left hand side is at least $v_S(S_i)/2$, which in turn means that either $S_i\cap C \in \mathcal{X}_S(C,P)$, or $S_i \setminus C \in \mathcal{X}_S({M \setminus C},P)$ (or both). So, $\lvert {X}_S(C,P) \rvert +\lvert \mathcal{X}_S({M \setminus C},P) \rvert \ge r$ and hence, $\lvert \mathcal{X}^*_{S}(C,P)\rvert \ge r/2 \ge \lceil r/2\rceil$.   
\end{proof}

 \section{Subadditive Valuations}
\label{sec:SubaddVal}

In this section we present our main technical result which establishes the existence of $1/2$-MMS allocation for the case of at most four agents (Theorem~\ref{thm:mainTheorem}). In Sections~\ref{sec:2Subadd}-~\ref{sec:4Subadd} we establish the existence of $1/2$-MMS approximate allocations for two, three and four subadditive agents (Corollaries~\ref{cor:2agents},~\ref{cor:3agents}~and~\ref{cor:4agents}), respectively, that follow as simple corollaries from more restricted settings (Lemmas~\ref{Lemma:2agents}, \ref{Lemma:threehalfs}, and \ref{Lemma:4agents}, respectively). We note that all these results improve the state-of-the-art for MMS guarantees for a wide set of valuation classes that lie below subadditive in the complement-free hierarchy. Subsequently, in Section \ref{sec:MoreAgents} we show how our proof techniques can be extended to obtain positive results for settings with many agents, when the agents have one of two types of admissible valuation functions. We emphasize that all the results presented herein are tight for subadditive valuations, i.e., there exist constructions where not all agents can receive more than $1/2$ of their MMS values \cite{GhodsiHSSY22}. In the proofs which follow we use ``symmetric'' cuts, i.e. no matter if $\mathcal{X}^*_i(C)=\mathcal{X}_i(C)$ or $\mathcal{X}^*_i(C)=\mathcal{X}_i(M \setminus C)$ for the maximum desired half of agent $i$ over cut $C$, the proof will continue the same way.

\begin{theorem}
\label{thm:mainTheorem}
    An $1/2$-MMS allocation exists for at most four agents with subadditive valuations.
\end{theorem}

Before proceeding with the proofs for few agents, we give the following general observation that illustrates the implications of our results. 

\begin{observation}
\label{obs:simpleReduction}
    Given some fair division instance, an allocation that is $\boldsymbol{\alpha}$-MMS$(\mathbf{d})$ is also $\boldsymbol{\alpha}'$-MMS$(\mathbf{d}')$, where $\boldsymbol{\alpha}$ is pointwise larger or equal to $\boldsymbol{\alpha}'$, and $\mathbf{d}$ is pointwise smaller or equal to $\mathbf{d}'$.
\end{observation}
\begin{proof}
    Let $A=(A_1,\ldots, A_n)$ be a $\boldsymbol{\alpha}$-MMS$(\mathbf{d})$ allocation. By definition it holds that for each agent $i$, $v_i(A_i)\geq \alpha_i \mu_i^{d_i}$. 
    
    We first show that $A$ is also a $\boldsymbol{\alpha}$-MMS$(\mathbf{d}')$ allocation.  
    Note that if for some agent $i$ there exists a partition of $M$ into $d_i'$ bundles that she values each by at least $\mu_i^{d_i'}$, then there exists a partition into $d_i\le d_i'$ bundles that she values by at least the same amount, by merging bundles of the first partition, due to monotonicity of the valuation functions. This in turns implies that $\mu_i^{d_i}\geq \mu_i^{d_i'}$. Therefore, $v_i(A_i)\geq \alpha_i \mu_i^{d_i'}$, for all $i$, which implies that $A$ is a $\boldsymbol{\alpha}$-MMS$(\mathbf{d}')$ allocation. 
    
    To complete the proof, it simply holds that $v_i(A_i)\geq \alpha_i \mu_i^{d_i'}\ge \alpha_i' \mu_i^{d_i'}$, for all $i$, which means that $A$ is also a $\boldsymbol{\alpha}'$-MMS$(\mathbf{d}')$ allocation. 
\end{proof}

\subsection{Two Agents}
\label{sec:2Subadd}
In this section we show that there always exists a $1/2$-MMS allocation for the case of two agents. We first show a stronger statement (Lemma~\ref{Lemma:2agents}) and obtain the main result as a corollary (Corollary~\ref{cor:2agents}). We further provide a useful restatement (Corollary~\ref{Cor:cut-and-choose}) of Lemma~\ref{Lemma:2agents} to be extensively used in the proofs for three and four agents.

\begin{lemma}[Two agents]
\label{Lemma:2agents}
    A $(1/2,1)$-MMS$(1,2)$ allocation exists for two agents with subadditive valuation functions. 
\end{lemma}

\begin{proof}

    We denote by $S$ and $T$ the two agents and by $S_j, T_j$ their $j$-th MMS bundle respectively, i.e., $S_1$ for the first agent and $T_1, T_2$ for the second agent. 
    We proceed in a cut-and-choose fashion; we partition the set of items into two disjoint bundles $T_1, T_2$ which both are worth at least $1$ for $T$. Agent $S$ picks her favorite bundle which (due to subadditivity) is guaranteed to have at least $1/2$ value since $v_S(T_1) + v_S(T_2) \geq v_S(T_1 \cup T_2)=v_S(S_1)= 1$, and $T$ receives the remaining bundle. Hence a $(1/2,1)$-MMS$(1,2)$ allocation exists (see Figure \ref{fig:2agentsRed} for an illustration).
\end{proof}

In the following corollary we state that Lemma~\ref{Lemma:2agents} suggests that a $1/2$-MMS approximation is attainable for two agents (by Observation~\ref{obs:simpleReduction}); this bound is tight for two agents \cite{GhodsiHSSY22}. 

\begin{corollary}
\label{cor:2agents}
    A $1/2$-MMS allocation exists for two agents with subadditive valuation functions.
\end{corollary}

\begin{center}
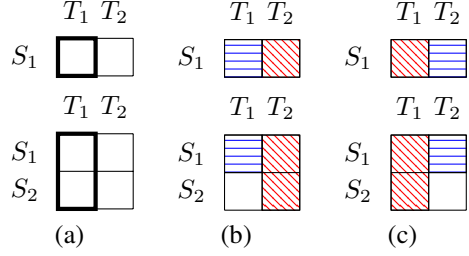
\begin{figure}
\begin{center}
    \begin{tabular}{ccc}
\begin{tikzpicture}[scale=0.5]
  \draw[step=1cm,] (0,0) grid (2,1);
  %\fill[black!20](0,0) rectangle +(1,1);
  \draw[black, ultra thick](0,0) rectangle +(1,1);
\node[anchor=north] at (0.5,2.25) {$T_1$};
\node[anchor=north] at (1.5,2.25) {$T_2$};
\node[anchor=west] at (-1.5,0.5) {$S_1$};
\end{tikzpicture}
& 
\begin{tikzpicture}[scale=0.5]
  \draw[step=1cm,] (0,0) grid (2,1);

(1,0) rectangle +(1,1);
\draw[pattern={north west lines},pattern color=red]
(1,0) rectangle +(1,1);
\draw[pattern={horizontal lines},pattern color=blue](0,0) rectangle (1,1);
\node[anchor=north] at (0.5,2.25) {$T_1$};
\node[anchor=north] at (1.5,2.25) {$T_2$};
\node[anchor=west] at (-1.5,0.5) {$S_1$};
\end{tikzpicture} 
& 
\begin{tikzpicture}[scale=0.5]
  \draw[step=1cm,] (0,0) grid (2,1);
\draw[pattern={north west lines},pattern color=red](0,0) rectangle (1,1);
\draw[pattern={horizontal lines},pattern color=blue](1,0) rectangle +(1,1);
\node[anchor=north] at (0.5,2.25) {$T_1$};
\node[anchor=north] at (1.5,2.25) {$T_2$};
\node[anchor=west] at (-1.5,0.5) {$S_1$};
\end{tikzpicture}\\
\begin{tikzpicture}[scale=0.5]

%\draw[pattern={north west lines},pattern color=blue](0,0) rectangle +(1,2);
\draw[black, ultra thick](0,0) rectangle +(1,2);
%\draw[pattern={horizontal lines},pattern color=red](1,0) rectangle (2,2);
%\draw[pattern={horizontal lines},pattern color=blue](1,0) rectangle (2,2);
  \draw[step=1cm,] (0,0) grid (2,2);

\node[anchor=north] at (0.5,3.25) {$T_1$};
\node[anchor=north] at (1.5,3.25) {$T_2$};
\node[anchor=west] at (-1.5,1.5) {$S_1$};
\node[anchor=west] at (-1.5,0.5) {$S_2$};
\end{tikzpicture} & \begin{tikzpicture}[scale=0.5]

(0,1) rectangle +(1,1);
\draw[pattern={horizontal lines},pattern color=blue]
(0,1) rectangle +(1,1);
;
\draw[pattern={north west lines},pattern color=red](1,0) rectangle (2,2);
  \draw[step=1cm,] (0,0) grid (2,2);

\node[anchor=north] at (0.5,3.25) {$T_1$};
\node[anchor=north] at (1.5,3.25) {$T_2$};
\node[anchor=west] at (-1.5,1.5) {$S_1$};
\node[anchor=west] at (-1.5,0.5) {$S_2$};
\end{tikzpicture} &
\begin{tikzpicture}[scale=0.5]
(0,0) rectangle +(1,2);
\draw[pattern={north west lines},pattern color=red]
(0,0) rectangle +(1,2);
\draw[pattern={horizontal lines},pattern color=blue](1,1) rectangle +(1,1);
  \draw[step=1cm,] (0,0) grid (2,2);

\node[anchor=north] at (0.5,3.25) {$T_1$};
\node[anchor=north] at (1.5,3.25) {$T_2$};
\node[anchor=west] at (-1.5,1.5) {$S_1$};
\node[anchor=west] at (-1.5,0.5) {$S_2$};
\end{tikzpicture}\\
(a)&(b)&(c)\\

    \end{tabular}
    \end{center}
        \caption{We illustrate the partition for both $\boldsymbol{d}=(1,2)$ and $\boldsymbol{d}=(2,2)$. The set of items $S_1$ can be divided into $T_1$ and $T_2$. If agent $S$ values bundle $T_1 \cap S_1$ (represented with a thick line in (a)) more than $\frac{v_S(S_1)}{2}$ then allocation (b) has the desired properties (the blue bundle for agent $S$ and the red bundle for agent $T$). If this is not the case, then due to the subadditivity, the same holds for agent $S$ and bundle $T_2 \cap S_1$; then (b) has the desired properties.} 
    \label{fig:2agentsRed}
\end{figure}
\end{center}
We also provide a useful restatement of Lemma~\ref{Lemma:2agents} to be used as a reduction tool in the proofs with more agents.
\begin{corollary}
\label{Cor:cut-and-choose}
    Consider a fair division instance of two agents with subadditive valuations and a set of $M$ items, where one agent values $M$ at least $1$, and there exists a partition into two bundles where the other agent values each at least $1/2$. Then, there exists an allocation that guarantees at least $1/2$ value to each agent.
\end{corollary}

\subsection{Three Agents}
\label{sec:3Subadd}
In this section we show that there exists a $1/2$-MMS allocation for the case of three agents. Again, we first show a stronger statement (Lemma~\ref{Lemma:threehalfs}) and obtain the main result as a corollary (Corollary~\ref{cor:3agents}). 

\begin{lemma} [Three agents]
\label{Lemma:threehalfs}
    An $1/2$-MMS$(3,2,2)$ allocation exists for three agents with subadditive valuation functions.
\end{lemma}
\begin{proof}
We denote by $S,T,Q$ the three agents and by $S_j,T_j,Q_j$ their $j$-th MMS bundle, respectively. 

The proof relies on the existence of $(1/2,1)$-MMS$(1,2)$ in the instance with two agents; we show that we can identify a valuable subset $A_T$ (with value at least $1/2$) to allocate to agent $T$, such that we can extend the allocation applying Corollary \ref{Cor:cut-and-choose} for agents $S$ and $Q$ and for the remaining items $M\setminus A_T$. In particular, agent $Q$ will have value at least $1$ for $M\setminus A_T$, while agent $S$ will be able to partition it into two bundles with value at least $1/2$ each. To achieve this we will use two cuts sequentially, first to agent $S$ and then (based on the response of $S$) to agent $T$  (we refer the reader to Figure~\ref{fig:3agents} (a) for an illustration of the proof).
    
{\bf First cut.} Consider the cut $C=T_2$ that we offer to $S$ and let $S_1^*,S_2^* \in \mathcal{X}_{S}^*(C)$ be two bundles in the maximum desired half (by Observation~\ref{obs:Cut_noOfDesiredSets} there exist at least two such bundles). The cut is ``symmetric'' for $T$ in a sense that both $C=T_2$ and $M\setminus C=T_1$ contain the {\em same number of $T$'s MMS bundles}; so it is without loss of generality to assume that $\mathcal{X}^*_S(C)=\mathcal{X}_S(C)$ (in Figure~\ref{fig:3agents} (a) $S_1^* \subseteq S_1,S_2^* \subseteq S_2$ are represented with blue color).

{\bf Second cut.} We next offer $C=Q_1$ to $T$ and let  $T_1^* \in \mathcal{X}_T^*(C)$ be the bundle in the maximum desired half. The cut is ``symmetric'' for $Q$, so again without loss of generality assume that $\mathcal{X}_T^*(C)=\mathcal{X}_T(C)$ (in Figure~\ref{fig:3agents} (b) $T_1^* \subseteq T_1$ is represented with red color).

{\bf Apply Corollary \ref{Cor:cut-and-choose}.} Overall, $v_T(T^*_1)\geq 1/2$ and $T^*_1$ will be allocated to $T$. Then, for $M'=M\setminus T^*_1$, it holds that $v_Q(M')\geq v_Q(Q_2) \geq 1$, and $S_1^*,S_2^*\subseteq M'$, for both of which $S$ has value at least $1/2$. So the lemma follows by using Corollary \ref{Cor:cut-and-choose} on $M'$ for agents $Q$ and $S$. (in Figure~\ref{fig:3agents} (c) $S_1^* \subseteq S_1,S_1^* \subseteq S_2$ are represented with blue color and $Q_2$ with green).
\end{proof}
\begin{center}

\begin{figure}
\begin{center}

\begin{tabular}{ccc}

\small{
\begin{tikzpicture}[scale=0.45]
\draw[yslant=0.5,xslant=-1,black, ultra thick](3,3) (5,1) rectangle +(-2,-1);  
\draw[yslant=0.5,black, ultra thick](3,-3) rectangle +(2,3);    \draw[yslant=-0.5,black, ultra thick](3,3) rectangle +(-1,-3); 
\draw[yslant=0.5,xslant=-1,pattern={horizontal lines},pattern color=blue] (5,1) rectangle +(-2,-1);  
\draw[yslant=0.5,pattern={horizontal lines},pattern color=blue](3,-2) rectangle +(2,2);    \draw[yslant=-0.5,pattern={horizontal lines},pattern color=blue](3,3) rectangle +(-1,-2);     

    \draw[yslant=-0.5] (1,0) grid (3,3);
  \draw[yslant=0.5] (3,-3) grid (5,0);
  \draw[yslant=0.5,xslant=-1] (3,0) grid (5,2);
\node[anchor=south west] at (-0.2,1.5,0) {$S_1$};
\node[anchor=south west] at (-0.2,0.5,0) {$S_2$};
\node[anchor=south west] at (-0.2,-0.5,0) {$S_3$};
\node[anchor=north west] at (0.45,3.7,0) {$Q_1$};
\node[anchor=north west] at (1.5,4.3,0) {$Q_2$};
\node[anchor=north west] at (0.5,-0.8,0) {$T_1$};
\node[anchor=north west] at (1.5,-1.3,0) {$T_2$};
\end{tikzpicture}} &
\small{
\begin{tikzpicture}[scale=0.45]
  
    \draw[yslant=0.5,xslant=-1,color=black,ultra thick] (4,2) rectangle +(-1,-2);
    %\fill[yslant=0.5,xslant=-1,color=black!20](4,2) rectangle +(-1,-2); 
    \draw[yslant=0.5,color=black,ultra thick] (3,-3) rectangle +(1,3);
    %\fill[yslant=0.5,color=black!20](3,-3) rectangle +(1,3);

    \draw[yslant=-0.5,color=black,ultra thick](3,3) rectangle +(-2,-3);
    %\fill[yslant=-0.5,color=black!20](3,3) rectangle +(-2,-3);
    
    \draw[yslant=0.5,xslant=-1,pattern={north west lines},pattern color=red](4,2) rectangle +(-1,-1);
\draw[yslant=-0.5,pattern={north west lines},pattern color=red](2,3) rectangle +(-1,-3);
  \draw[yslant=-0.5] (1,0) grid (3,3);
  \draw[yslant=0.5] (3,-3) grid (5,0);
  \draw[yslant=0.5,xslant=-1] (3,0) grid (5,2);
\node[anchor=south west] at (-0.2,1.5,0) {$S_1$};
\node[anchor=south west] at (-0.2,0.5,0) {$S_2$};
\node[anchor=south west] at (-0.2,-0.5,0) {$S_3$};
\node[anchor=north west] at (0.45,3.7,0) {$Q_1$};
\node[anchor=north west] at (1.5,4.3,0) {$Q_2$};
\node[anchor=north west] at (0.5,-0.8,0) {$T_1$};
\node[anchor=north west] at (1.5,-1.3,0) {$T_2$};
\end{tikzpicture}} &
\small{
\begin{tikzpicture}[scale=0.45]
 \draw[yslant=0.5,xslant=-1,pattern={north west lines},pattern color=red](4,2) rectangle +(-1,-1);
\draw[yslant=-0.5,pattern={north west lines},pattern color=red](2,3) rectangle +(-1,-3);
\draw[yslant=0.5,xslant=-1,pattern={crosshatch dots},pattern color=green]  (5,2) rectangle +(-1,-2);
    \draw[yslant=0.5,pattern={crosshatch dots},pattern color=green]  (4,-3) rectangle +(1,3);

    \draw[yslant=0.5,xslant=-1,pattern={horizontal lines},pattern color=blue] (5,1) rectangle +(-2,-1);
    \draw[yslant=0.5,pattern={horizontal lines},pattern color=blue](3,-2) rectangle +(2,2);
    \draw[yslant=-0.5,pattern={horizontal lines},pattern color=blue](3,3) rectangle (2,1);
  \draw[yslant=-0.5] (1,0) grid (3,3);
  \draw[yslant=0.5] (3,-3) grid (5,0);
  \draw[yslant=0.5,xslant=-1] (3,0) grid (5,2);
\node[anchor=south west] at (-0.2,1.5,0) {$S_1$};
\node[anchor=south west] at (-0.2,0.5,0) {$S_2$};
\node[anchor=south west] at (-0.2,-0.5,0) {$S_3$};
\node[anchor=north west] at (0.45,3.7,0) {$Q_1$};
\node[anchor=north west] at (1.5,4.3,0) {$Q_2$};
\node[anchor=north west] at (0.5,-0.8,0) {$T_1$};
\node[anchor=north west] at (1.5,-1.3,0) {$T_2$};
\end{tikzpicture}}\\
(a) & (b) & (c) \\
$S_1^*,S_2^* \in \mathcal{X}_S^*(T_2)$ &  $T_1^* \in \mathcal{X}^*_T(Q_1)$ & Apply Corollary \ref{Cor:cut-and-choose}\\
\end{tabular}
    
\end{center}
\caption{We use blue, red, and green to denote the bundles from which we will allocate to agents $S,T$ and $Q$, respectively. The thickened lines illustrate the first and second cuts, while (c) demonstrates the application of Corollary 2.
}
\label{fig:3agents}
\end{figure}
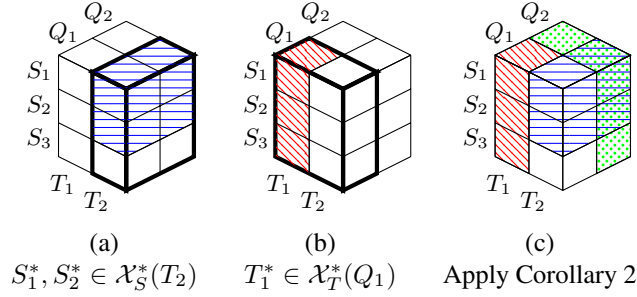 
\end{center}
As a corollary of Lemma~\ref{Lemma:threehalfs}, a $1/2$-MMS allocation always exists for three agents (by Observation~\ref{obs:simpleReduction}); this bound is also tight \cite{GhodsiHSSY22}. 
\begin{corollary}
\label{cor:3agents}
    A $1/2$-MMS allocation exists for three agents with subadditive valuation functions. 
\end{corollary}

\subsection{Four Agents}
In this section we show the existence of $1/2$-MMS allocation for the case of four agents, our main technical result. We are also able to show a stronger statement (Lemma~\ref{Lemma:4agents}) and obtain the main result as a corollary (Corollary~\ref{cor:4agents}). 

\label{sec:4Subadd}

\begin{lemma} [Four agents]
\label{Lemma:4agents}
    A $1/2$-MMS$(3,3,4,4)$ allocation exists for four agents with subadditive valuation functions.
\end{lemma}
\begin{proof}
    We denote by $S,T,Q,R$ the four agents and by $S_j,T_j,Q_j,R_j$ their $j$-th MMS bundle, respectively. 
    
    We progressively identify four candidate allocations in total, and we show that one of those should be a $1/2$-MMS$(3,3,4,4)$ allocation. 

    In a nutshell, the protocol works as follows. By offering carefully chosen cuts to agents $S,T,Q$ we are able to find partial allocations for two of those agents (in the first round these are $S$ and $T$), that have high enough value (higher than $1/2$), in a way that always preserves one MMS bundle of agent $R$ {\em intact}, let it be $R_4$. If the third agent (this is $Q$ in the first round) values higher than $1/2$ two of her bundles (on the remaining items) then we can apply Corollary~\ref{Cor:cut-and-choose} on $Q$ and $R$ and for the remaining items, and we can claim a $1/2$-MMS$(3,3,4,4)$ allocation. If this is not the case, then the third agent must value at least two of the intersections of her bundles with the partial allocation of the other two agents with value higher than $1/2$. In the next round we will tentatively allocate to the third agent one of these sets. We will also keep one of the other two agents and offer her a subset of high value {\em from a different bundle} (the notion of maximum desired half is useful to achieve this last property). We proceed in a similar manner by querying the remaining agent investigating again whether Corollary~\ref{Cor:cut-and-choose} can be employed. The key is that in every round that the third agent does not satisfy the conditions of Corollary~\ref{Cor:cut-and-choose} we are able to build a more structured partial allocation, which in the final step provides an allocation of $M\setminus R_4$ to agents $S,T$ and $Q$, that they value by at least $1/2$. Then $R_4$ can be allocated to agent $R$ and the allocation is $1/2$-MMS$(3,3,4,4)$.
    
\begin{center}

\begin{figure}[t]
\begin{center}

    \begin{tabular}{cc}
\small{
\begin{tikzpicture}[scale=0.425]
\draw[yslant=0.5,ultra thick, black](4,-3) rectangle +(4,1);     
\draw[yslant=-0.5,ultra thick, black](4,2) rectangle +(-3,-1); 
\draw[yslant=-0.5,pattern={north west lines},pattern color=red](4,3) rectangle +(-1,-1);
\draw[yslant=-0.5,pattern={north west lines},pattern color=red](4,1) rectangle +(-1,-1);
\draw[yslant=0.5,xslant=-1,pattern={north west lines},pattern color=red](7,0) rectangle +(-4,-1); 
\draw[yslant=0.5,pattern={north west lines},pattern color=red](4,-2) rectangle +(4,1);   
\draw[yslant=0.5,pattern={north west lines},pattern color=red](4,-4) rectangle +(4,1);   
%\draw[yslant=0.5,xslant=-1,pattern={horizontal lines},pattern color=blue](7,2) rectangle +(-4,-3);  
\draw[yslant=0.5,pattern={horizontal lines},pattern color=blue](4,-3) rectangle +(4,1);     
\draw[yslant=-0.5,pattern={horizontal lines},pattern color=blue](4,2) rectangle +(-3,-1);     
  \draw[yslant=-0.5] (1,0) grid (4,3);
  \draw[yslant=0.5] (4,-4) grid (8,-1);
  \draw[yslant=0.5,xslant=-1] (3,-1) grid (7,2);
\node[anchor=south west] at (-0.25,1.5,0) {$S_1$};
\node[anchor=south west] at (-0.25,0.5,0) {$S_2$};
\node[anchor=south west] at (-0.25,-0.5,0) {$S_3$};
\node[anchor=north west] at (0.5,3.8,0) {$Q_1$};
\node[anchor=north west] at (1.5,4.3,0) {$Q_2$};
\node[anchor=north west] at (2.5,4.8,0) {$Q_3$};
\node[anchor=north west] at (3.5,5.3,0) {$Q_4$};
\node[anchor=north west] at (0.4,-0.8,0) {$T_3$};
\node[anchor=north west] at (1.4,-1.3,0) {$T_2$};
\node[anchor=north west] at (2.4,-1.8,0) {$T_1$};
\node[anchor=north] at (7,5,0) {$R_1$};
\end{tikzpicture}}
&
\small{
\begin{tikzpicture}[scale=0.425]
\draw[yslant=0.5,ultra thick, black](4,-3) rectangle +(4,1);     
\draw[yslant=-0.5,ultra thick, black](4,2) rectangle +(-3,-1); 
\draw[yslant=-0.5,pattern={north west lines},pattern color=red](4,3) rectangle +(-1,-1);
\draw[yslant=-0.5,pattern={north west lines},pattern color=red](4,1) rectangle +(-1,-1);
\draw[yslant=0.5,xslant=-1,pattern={north west lines},pattern color=red](7,0) rectangle +(-4,-1); 
\draw[yslant=0.5,pattern={north west lines},pattern color=red](4,-2) rectangle +(4,1);   
\draw[yslant=0.5,pattern={north west lines},pattern color=red](4,-4) rectangle +(4,1);   
%\draw[yslant=0.5,xslant=-1,pattern={horizontal lines},pattern color=blue](7,2) rectangle +(-4,-3);  
\draw[yslant=0.5,pattern={horizontal lines},pattern color=blue](4,-3) rectangle +(4,1);     
\draw[yslant=-0.5,pattern={horizontal lines},pattern color=blue](4,2) rectangle +(-3,-1);      
  \draw[yslant=-0.5] (1,0) grid (4,3);
  \draw[yslant=0.5] (4,-4) grid (8,-1);
  \draw[yslant=0.5,xslant=-1] (3,-1) grid (7,2);
\node[anchor=south west] at (-0.25,1.5,0) {$S_1$};
\node[anchor=south west] at (-0.25,0.5,0) {$S_2$};
\node[anchor=south west] at (-0.25,-0.5,0) {$S_3$};
\node[anchor=north west] at (0.5,3.8,0) {$Q_1$};
\node[anchor=north west] at (1.5,4.3,0) {$Q_2$};
\node[anchor=north west] at (2.5,4.8,0) {$Q_3$};
\node[anchor=north west] at (3.5,5.3,0) {$Q_4$};
\node[anchor=north west] at (0.4,-0.8,0) {$T_3$};
\node[anchor=north west] at (1.4,-1.3,0) {$T_2$};
\node[anchor=north west] at (2.4,-1.8,0) {$T_1$};
\node[anchor=north] at (7,5,0) {$R_2$};
\end{tikzpicture}}\\
\small{
\begin{tikzpicture}[scale=0.425]
\draw[yslant=-0.5,pattern={north west lines},pattern color=red](4,3) rectangle +(-1,-3);
\draw[yslant=0.5,xslant=-1,pattern={north west lines},pattern color=red](7,0) rectangle +(-4,-1); 
\draw[yslant=0.5,pattern={north west lines},pattern color=red](4,-4) rectangle +(4,3);    
  \draw[yslant=-0.5] (1,0) grid (4,3);
  \draw[yslant=0.5] (4,-4) grid (8,-1);
  \draw[yslant=0.5,xslant=-1] (3,-1) grid (7,2);
\node[anchor=south west] at (-0.25,1.5,0) {$S_1$};
\node[anchor=south west] at (-0.25,0.5,0) {$S_2$};
\node[anchor=south west] at (-0.25,-0.5,0) {$S_3$};
\node[anchor=north west] at (0.5,3.8,0) {$Q_1$};
\node[anchor=north west] at (1.5,4.3,0) {$Q_2$};
\node[anchor=north west] at (2.5,4.8,0) {$Q_3$};
\node[anchor=north west] at (3.5,5.3,0) {$Q_4$};
\node[anchor=north west] at (0.4,-0.8,0) {$T_3$};
\node[anchor=north west] at (1.4,-1.3,0) {$T_2$};
\node[anchor=north west] at (2.4,-1.8,0) {$T_1$};
\node[anchor=north] at (7,5,0) {$R_3$};
\end{tikzpicture}}
&

\small{
\begin{tikzpicture}[scale=0.425]
\draw[yslant=0.5,ultra thick, black](4,-4) rectangle +(4,3);     
\draw[yslant=-0.5,ultra thick, black](4,3) rectangle +(-3,-3); 
\draw[yslant=0.5,xslant=-1,ultra thick, black](7,2) rectangle +(-4,-3);    
  \draw[yslant=-0.5] (1,0) grid (4,3);
  \draw[yslant=0.5] (4,-4) grid (8,-1);
  \draw[yslant=0.5,xslant=-1] (3,-1) grid (7,2);
\node[anchor=south west] at (-0.25,1.5,0) {$S_1$};
\node[anchor=south west] at (-0.25,0.5,0) {$S_2$};
\node[anchor=south west] at (-0.25,-0.5,0) {$S_3$};
\node[anchor=north west] at (0.5,3.8,0) {$Q_1$};
\node[anchor=north west] at (1.5,4.3,0) {$Q_2$};
\node[anchor=north west] at (2.5,4.8,0) {$Q_3$};
\node[anchor=north west] at (3.5,5.3,0) {$Q_4$};
\node[anchor=north west] at (0.4,-0.8,0) {$T_3$};
\node[anchor=north west] at (1.4,-1.3,0) {$T_2$};
\node[anchor=north west] at (2.4,-1.8,0) {$T_1$};
\node[anchor=north] at (7,5,0) {$R_4$};
\end{tikzpicture}}\\
\end{tabular}
    
\end{center}
\caption{The first candidate allocation $A=(S_2^*,T_1^*)$ for four agents and $\boldsymbol{d}=(3,3,4,4)$. We use blue to denote $S_2^* \subseteq S_2$ and red to denote $T_1^* \subseteq T_1$. We use a thick line to illustrate the cut $C=\{S^*_2 \cup R_4\}$. The allocation is valid and none of the bundles intersects with $R_4$. The cuts are symmetric, i.e. if we had $\mathcal{X}_S^*(C) =\mathcal{X}_S^*(M\setminus C)$ for cut $C=\{R_1 \cup R_2\}$ or/and $\mathcal{X}_T^*(C)=\mathcal{X}_T(C)$ for the corresponding cut $C$, we could construct the same allocation by renaming the bundles.}
\label{fig:4agents1}
\end{figure}
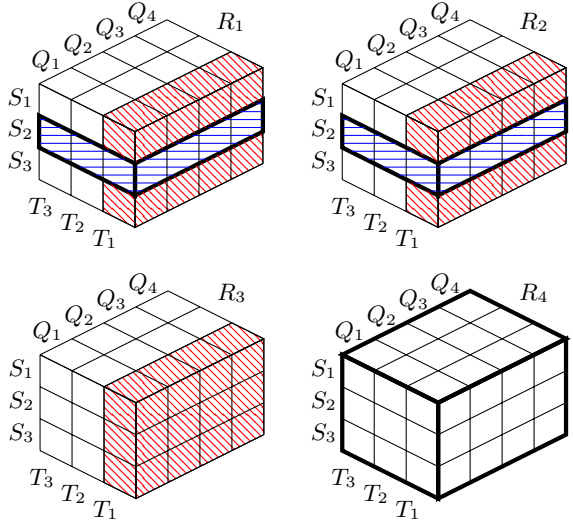
\end{center}

    \noindent{\bf Building the first candidate allocation}. We first consider agent $S$ (an agent with $3$ MMS bundles) and cut her MMS bundles by offering the cut $C=R_1 \cup R_2$. By Observation~\ref{obs:Cut_noOfDesiredSets} there are at least two bundles in the set $\mathcal{X}_{S}^*(C)$; let those be $S_1^*,S_2^* \in \mathcal{X}_{S}^*(C)$. We note that $S_1^*,S_2^*$ intersect with exactly two MMS bundles of $R$, so w.l.o.g. assume that those are $R_1,R_2$, i.e., assume that $\mathcal{X}_{S}^*(C)=\mathcal{X}_{S}(C)$. Therefore
     \begin{equation}
        \label{eq:4agentsCond0}
        S_j^*\cap R_3 = \emptyset \mbox{ and } S_j^*\cap R_4 = \emptyset, \mbox{ for } j \in \{1,2\}\,.
    \end{equation}

    Next, we offer a cut to agent $T$ in a way that a) ensures that a {\em whole} MMS bundle of agent $R$ remains intact and b) one of $S_1^*, S_2^*$ does not intersect with the Maximum Desired Half of $T$. A cut that serves this purpose is $C=S_2^* \cup R_4$. Now each of the sets $C$ and $M\setminus C$ intersects with exactly one of $S^*_1, S^*_2$ and with exactly one of $R_3, R_4$ which are the remaining whole bundles of $R$. W.l.o.g. assume that $\mathcal{X}_{T}^*(C)=\mathcal{X}_{T}(M \setminus C)$.\footnote{Even if it is w.l.o.g., we consider $\mathcal{X}_{T}^*(C)=\mathcal{X}_{T}(M \setminus C)$ that includes $S_3$ to avoid any confusion of the steps needed, because the case of $\mathcal{X}_{T}^*(C)=\mathcal{X}_{T}(C)$ is simpler and can be handled the same way.} 
    By Observation~\ref{obs:Cut_noOfDesiredSets} there are at least two bundles in the set $\mathcal{X}_{T}^*(C)$, let those be $T_1^*,T_2^* \in \mathcal{X}_{T}^*(C)$. Then, it holds that 
    \begin{equation}
    \label{eq:4agentsCond1}
        T_j^*\cap R_4 = \emptyset \mbox{ and } T_j^*\cap S_2^* = \emptyset, \mbox{ for } j \in \{1,2\}\,.
    \end{equation}
    
    {\bf First candidate allocation.} Consider the following partial allocation for $S$ and $T$: $A= (S^*_2,T^*_1)$ (see Figure \ref{fig:4agents1} for an illustration). By construction, both $S$ and $T$ value their allocated bundles by at least $1/2$. If there exist at least two MMS bundles of $Q$ that she values with at least $1/2$ {\em after the removal of $S^*_2\cup T^*_1$}, then the conditions of Corollary~\ref{Cor:cut-and-choose} are satisfied for $Q$ and $R$ for the remaining items (recall that $R$ values the remaining items by at least $1$ since they contain $R_4$). Hence by  employing Corollary~\ref{Cor:cut-and-choose} we can find an allocation of $M\setminus (S^*_2\cup T^*_1)$ to $Q$ and $R$ where they both value their bundles with at least $1/2$, and we are done.   
        
    So, suppose that this is not the case. Then there must be at least three MMS bundles of $Q$, let them be $Q_1,Q_2,Q_3$, such that $v_Q(Q_j\cap(S^*_2\cup T^*_1))\geq 1/2$ for $j\in\{1, 2, 3\}$. In other words, if we consider the cut $C=S^*_2\cup T^*_1$ for agent $Q$, then it is guaranteed that $Q_1^*,Q^*_2,Q^*_3 \in \mathcal{X}_{Q}(C)$. Since $Q_j^*\subseteq S^*_2\cup T^*_1$, for all $j \in \{1,2,3\}$ and also by \eqref{eq:4agentsCond0} and \eqref{eq:4agentsCond1} we conclude that
    \begin{equation}
        \label{eq:4agentsCond2}
        Q_j^* \cap T^*_2 = \emptyset \mbox{ and } Q_j^* \cap R_4= \emptyset , \mbox{ for } j \in \{1,2,3\}\,.
    \end{equation}

    \begin{figure}[h] 
\begin{center}
    \begin{tabular}{cc}
\small{
\begin{tikzpicture}[scale=0.425]
\draw[yslant=-0.5,ultra thick, black](4,3) rectangle +(-1,-3);
\draw[yslant=0.5,xslant=-1,ultra thick, black](7,0) rectangle +(-4,-1); 
\draw[yslant=0.5,ultra thick, black](4,-4) rectangle +(4,3);   
       
\draw[yslant=-0.5,ultra thick, black](3,2) rectangle +(-2,-1);
\draw[yslant=-0.5,ultra thick, white](3,2) rectangle +(0,-1);
\draw[yslant=0.5,pattern={crosshatch dots},pattern color=green] (4,-4) rectangle +(1,3);
\draw[yslant=-0.5,pattern={crosshatch dots},pattern color=green](4,3) rectangle +(-1,-3);
\draw[yslant=-0.5,pattern={crosshatch dots},pattern color=green](3,2) rectangle +(-2,-1);
\draw[yslant=-0.5,pattern={north west lines},pattern color=red](3,3) rectangle +(-1,-1);
\draw[yslant=-0.5,pattern={north west lines},pattern color=red](3,1) rectangle +(-1,-1);
\draw[yslant=0.5,xslant=-1,pattern={north west lines},pattern color=red](7,1) rectangle +(-4,-1); 
\draw[yslant=0.5,xslant=-1,pattern={crosshatch dots},pattern color=green](4,0) rectangle +(-1,-1);     
  \draw[yslant=-0.5] (1,0) grid (4,3);
  \draw[yslant=0.5] (4,-4) grid (8,-1);
  \draw[yslant=0.5,xslant=-1] (3,-1) grid (7,2);
\node[anchor=south west] at (-0.25,1.5,0) {$S_1$};
\node[anchor=south west] at (-0.25,0.5,0) {$S_2$};
\node[anchor=south west] at (-0.25,-0.5,0) {$S_3$};
\node[anchor=north west] at (0.5,3.8,0) {$Q_1$};
\node[anchor=north west] at (1.5,4.3,0) {$Q_2$};
\node[anchor=north west] at (2.5,4.8,0) {$Q_3$};
\node[anchor=north west] at (3.5,5.3,0) {$Q_4$};
\node[anchor=north west] at (0.4,-0.8,0) {$T_3$};
\node[anchor=north west] at (1.4,-1.3,0) {$T_2$};
\node[anchor=north west] at (2.4,-1.8,0) {$T_1$};
\node[anchor=north] at (7,5,0) {$R_1$};
\end{tikzpicture}}
&
\small{
\begin{tikzpicture}[scale=0.425]
\draw[yslant=-0.5,ultra thick, black](4,3) rectangle +(-1,-3);
\draw[yslant=0.5,xslant=-1,ultra thick, black](7,0) rectangle +(-4,-1); 
\draw[yslant=0.5,ultra thick, black](4,-4) rectangle +(4,3);   
       
\draw[yslant=-0.5,ultra thick, black](3,2) rectangle +(-2,-1);
\draw[yslant=-0.5,ultra thick, white](3,2) rectangle +(0,-1);
\draw[yslant=0.5,pattern={crosshatch dots},pattern color=green] (4,-4) rectangle +(1,3);
\draw[yslant=-0.5,pattern={crosshatch dots},pattern color=green](4,3) rectangle +(-1,-3);
\draw[yslant=-0.5,pattern={crosshatch dots},pattern color=green](3,2) rectangle +(-2,-1);
\draw[yslant=-0.5,pattern={north west lines},pattern color=red](3,3) rectangle +(-1,-1);
\draw[yslant=-0.5,pattern={north west lines},pattern color=red](3,1) rectangle +(-1,-1);
\draw[yslant=0.5,xslant=-1,pattern={north west lines},pattern color=red](7,1) rectangle +(-4,-1); 
\draw[yslant=0.5,xslant=-1,pattern={crosshatch dots},pattern color=green](4,0) rectangle +(-1,-1);    
  \draw[yslant=-0.5] (1,0) grid (4,3);
  \draw[yslant=0.5] (4,-4) grid (8,-1);
  \draw[yslant=0.5,xslant=-1] (3,-1) grid (7,2);
\node[anchor=south west] at (-0.25,1.5,0) {$S_1$};
\node[anchor=south west] at (-0.25,0.5,0) {$S_2$};
\node[anchor=south west] at (-0.25,-0.5,0) {$S_3$};
\node[anchor=north west] at (0.5,3.8,0) {$Q_1$};
\node[anchor=north west] at (1.5,4.3,0) {$Q_2$};
\node[anchor=north west] at (2.5,4.8,0) {$Q_3$};
\node[anchor=north west] at (3.5,5.3,0) {$Q_4$};
\node[anchor=north west] at (0.4,-0.8,0) {$T_3$};
\node[anchor=north west] at (1.4,-1.3,0) {$T_2$};
\node[anchor=north west] at (2.4,-1.8,0) {$T_1$};
\node[anchor=north] at (7,5,0) {$R_2$};
\end{tikzpicture}}\\
\small{
\begin{tikzpicture}[scale=0.425]
\draw[yslant=-0.5,ultra thick, black](4,3) rectangle +(-1,-3);
\draw[yslant=0.5,xslant=-1,ultra thick, black](7,0) rectangle +(-4,-1); 
\draw[yslant=0.5,ultra thick, black](4,-4) rectangle +(4,3);   
       
\draw[yslant=-0.5,pattern={north west lines},pattern color=red](3,3) rectangle +(-1,-3);
\draw[yslant=0.5,xslant=-1,pattern={north west lines},pattern color=red](7,1) rectangle +(-4,-1); 
\draw[yslant=-0.5,pattern={crosshatch dots},pattern color=green](4,3) rectangle +(-1,-3);
\draw[yslant=0.5,xslant=-1,pattern={crosshatch dots},pattern color=green](4,0) rectangle +(-1,-1); 
\draw[yslant=0.5,pattern={crosshatch dots},pattern color=green](4,-4) rectangle +(1,3);    
  \draw[yslant=-0.5] (1,0) grid (4,3);
  \draw[yslant=0.5] (4,-4) grid (8,-1);
  \draw[yslant=0.5,xslant=-1] (3,-1) grid (7,2);
\node[anchor=south west] at (-0.25,1.5,0) {$S_1$};
\node[anchor=south west] at (-0.25,0.5,0) {$S_2$};
\node[anchor=south west] at (-0.25,-0.5,0) {$S_3$};
\node[anchor=north west] at (0.5,3.8,0) {$Q_1$};
\node[anchor=north west] at (1.5,4.3,0) {$Q_2$};
\node[anchor=north west] at (2.5,4.8,0) {$Q_3$};
\node[anchor=north west] at (3.5,5.3,0) {$Q_4$};
\node[anchor=north west] at (0.4,-0.8,0) {$T_3$};
\node[anchor=north west] at (1.4,-1.3,0) {$T_2$};
\node[anchor=north west] at (2.4,-1.8,0) {$T_1$};
\node[anchor=north] at (7,5,0) {$R_3$};
\end{tikzpicture}}
&

\small{
\begin{tikzpicture}[scale=0.425]

  \draw[yslant=-0.5] (1,0) grid (4,3);
  \draw[yslant=0.5] (4,-4) grid (8,-1);
  \draw[yslant=0.5,xslant=-1] (3,-1) grid (7,2);
\node[anchor=south west] at (-0.25,1.5,0) {$S_1$};
\node[anchor=south west] at (-0.25,0.5,0) {$S_2$};
\node[anchor=south west] at (-0.25,-0.5,0) {$S_3$};
\node[anchor=north west] at (0.5,3.8,0) {$Q_1$};
\node[anchor=north west] at (1.5,4.3,0) {$Q_2$};
\node[anchor=north west] at (2.5,4.8,0) {$Q_3$};
\node[anchor=north west] at (3.5,5.3,0) {$Q_4$};
\node[anchor=north west] at (0.4,-0.8,0) {$T_3$};
\node[anchor=north west] at (1.4,-1.3,0) {$T_2$};
\node[anchor=north west] at (2.4,-1.8,0) {$T_1$};
\node[anchor=north] at (7,5,0) {$R_4$};
\end{tikzpicture}}\\
\end{tabular}
    
\end{center}
\caption{The second candidate allocation $A'=(T_2^*,Q_1^*)$. We use red for bundle $T_2^* \subseteq T_2$ and green for $Q_1^*\subseteq Q_1$. The cut $C=\{T_1^* \cup S_2^*\}$ is shown with a thick line. We try to apply Corollary \ref{cor:2agents} for agents $S$ and $R$ and the set of items $M \setminus (T_2^*\cup Q_1^*)$. We could construct the same allocation for any $3$ bundles $Q_i^*$ by renaming the bundles.}
\label{fig:4agents2}
\end{figure}
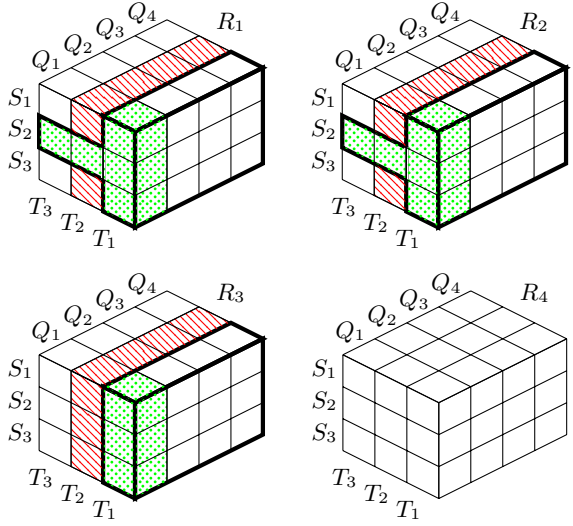

    {\bf Second candidate allocation.} Next, we consider the partial allocation $A'= (T^*_2,Q^*_1)$ for agents $T$ and $Q$ which by (\ref{eq:4agentsCond2}) is valid and both $v_T(T^*_2), v_Q(Q_1^*)$ are higher than $1/2$ (see Figure \ref{fig:4agents2} for an illustration). If there exist at least two MMS bundles of $S$, that she values with at least $1/2$ {\em after the removal of $T^*_2 \cup Q^*_1$}, then by employing Corollary \ref{Cor:cut-and-choose} we can find an allocation of $M\setminus (T^*_2\cup Q^*_1)$ to $S$ and $R$ where they both value their bundles with at least $1/2$, and we are done.

    Otherwise, it should be that for the cut $C=T^*_2 \cup Q^*_1$  there are two sets\footnote{We note that it is not necessarily the case that $S'_1\subseteq S_1$ or $S'_2\subseteq S_2$ although this is how it is depicted in the figures for the sake of exposition.} $S_1',S_2' \in \mathcal{X}_{S}(C)$. 
    Since for any $j \in \{1,2\}$, $S_j'\subseteq T^*_2 \cup Q^*_1$, by \eqref{eq:4agentsCond1} and \eqref{eq:4agentsCond2} we conclude
    \begin{equation}
        \label{eq:4agentsCond3}
        S_j' \cap Q^*_2 = \emptyset, S_j' \cap Q^*_3 = \emptyset \mbox{ and } S_j' \cap R_4 = \emptyset, \mbox{ for } j \in \{1,2\}\,.
    \end{equation}

    {\bf Third candidate allocation.} Consider now the partial allocation $A''= (S'_1,Q^*_2)$ for $S$ and $Q$, which is valid (due to (\ref{eq:4agentsCond3})) and both agents value their allocated bundles by at least $1/2$ (see Figure \ref{fig:4agents3} for an illustration).
    Again, if there exist at least two MMS bundles of $T$, that she values with at least $1/2$ after the removal of $S'_1\cup Q^*_2$, by using Corollary \ref{Cor:cut-and-choose} on $T$ and $R$, we are done (similarly as in the previous cases).

\begin{figure}[H]
\begin{center}

    \begin{tabular}{cc}
\small{
\begin{tikzpicture}[scale=0.425]

\draw[yslant=0.5,ultra thick, black] (4,-4) rectangle +(1,3);
\draw[yslant=-0.5,ultra thick, black](4,3) rectangle +(-2,-3);
\draw[yslant=-0.5,ultra thick, black](2,2) rectangle +(-1,-1);
\draw[yslant=-0.5,ultra thick, white](3,1) rectangle +(-1,1);
\draw[yslant=0.5,xslant=-1,ultra thick, black](7,1) rectangle +(-4,-1); 
\draw[yslant=0.5,xslant=-1,ultra thick, black](4,0) rectangle +(-1,-1); 
\draw[yslant=0.5,xslant=-1,ultra thick, white](3,0) rectangle +(1,0); 

\draw[yslant=0.5,pattern={crosshatch dots},pattern color=green] (5,-4) rectangle +(1,3);
\draw[yslant=0.5,xslant=-1,pattern={crosshatch dots},pattern color=green](5,0) rectangle +(-1,-1); 
\draw[yslant=0.5,pattern={horizontal lines},pattern color=blue](4,-2) rectangle +(1,1);    
\draw[yslant=0.5,xslant=-1,pattern={horizontal lines},pattern color=blue](7,1) rectangle +(-4,-1);
\draw[yslant=0.5,xslant=-1,pattern={horizontal lines},pattern color=blue](4,0) rectangle +(-1,-1);
%\draw[yslant=-0.5,pattern={horizontal lines},pattern color=blue](4,2) rectangle +(-3,-1);     
\draw[yslant=-0.5,pattern={horizontal lines},pattern color=blue](4,3) rectangle +(-2,-1);    
  \draw[yslant=-0.5] (1,0) grid (4,3);
  \draw[yslant=0.5] (4,-4) grid (8,-1);
  \draw[yslant=0.5,xslant=-1] (3,-1) grid (7,2);
\node[anchor=south west] at (-0.25,1.5,0) {$S_1$};
\node[anchor=south west] at (-0.25,0.5,0) {$S_2$};
\node[anchor=south west] at (-0.25,-0.5,0) {$S_3$};
\node[anchor=north west] at (0.5,3.8,0) {$Q_1$};
\node[anchor=north west] at (1.5,4.3,0) {$Q_2$};
\node[anchor=north west] at (2.5,4.8,0) {$Q_3$};
\node[anchor=north west] at (3.5,5.3,0) {$Q_4$};
\node[anchor=north west] at (0.4,-0.8,0) {$T_3$};
\node[anchor=north west] at (1.4,-1.3,0) {$T_2$};
\node[anchor=north west] at (2.4,-1.8,0) {$T_1$};
\node[anchor=north] at (7,5,0) {$R_1$};
\end{tikzpicture}}
&
\small{
\begin{tikzpicture}[scale=0.425]

\draw[yslant=0.5,ultra thick, black] (4,-4) rectangle +(1,3);
\draw[yslant=-0.5,ultra thick, black](4,3) rectangle +(-2,-3);
\draw[yslant=-0.5,ultra thick, black](2,2) rectangle +(-1,-1);
\draw[yslant=-0.5,ultra thick, white](3,1) rectangle +(-1,1);
\draw[yslant=0.5,xslant=-1,ultra thick, black](7,1) rectangle +(-4,-1); 
\draw[yslant=0.5,xslant=-1,ultra thick, black](4,0) rectangle +(-1,-1); 
\draw[yslant=0.5,xslant=-1,ultra thick, white](3,0) rectangle +(1,0); 

\draw[yslant=0.5,pattern={crosshatch dots},pattern color=green] (5,-4) rectangle +(1,3);

\draw[yslant=0.5,xslant=-1,pattern={crosshatch dots},pattern color=green](5,0) rectangle +(-1,-1); 
\draw[yslant=0.5,pattern={horizontal lines},pattern color=blue](4,-2) rectangle +(1,1);    
\draw[yslant=0.5,xslant=-1,pattern={horizontal lines},pattern color=blue](7,1) rectangle +(-4,-1);
\draw[yslant=0.5,xslant=-1,pattern={horizontal lines},pattern color=blue](4,0) rectangle +(-1,-1);
%\draw[yslant=-0.5,pattern={horizontal lines},pattern color=blue](4,2) rectangle +(-3,-1);     
\draw[yslant=-0.5,pattern={horizontal lines},pattern color=blue](4,3) rectangle +(-2,-1);  
  \draw[yslant=-0.5] (1,0) grid (4,3);
  \draw[yslant=0.5] (4,-4) grid (8,-1);
  \draw[yslant=0.5,xslant=-1] (3,-1) grid (7,2);
\node[anchor=south west] at (-0.25,1.5,0) {$S_1$};
\node[anchor=south west] at (-0.25,0.5,0) {$S_2$};
\node[anchor=south west] at (-0.25,-0.5,0) {$S_3$};
\node[anchor=north west] at (0.5,3.8,0) {$Q_1$};
\node[anchor=north west] at (1.5,4.3,0) {$Q_2$};
\node[anchor=north west] at (2.5,4.8,0) {$Q_3$};
\node[anchor=north west] at (3.5,5.3,0) {$Q_4$};
\node[anchor=north west] at (0.4,-0.8,0) {$T_3$};
\node[anchor=north west] at (1.4,-1.3,0) {$T_2$};
\node[anchor=north west] at (2.4,-1.8,0) {$T_1$};
\node[anchor=north] at (7,5,0) {$R_2$};
\end{tikzpicture}}\\
\small{
\begin{tikzpicture}[scale=0.425]
\draw[yslant=0.5,ultra thick, black] (4,-4) rectangle +(1,3);
\draw[yslant=-0.5,ultra thick, black](4,3) rectangle +(-2,-3);
\draw[yslant=0.5,xslant=-1,ultra thick, black](7,1) rectangle +(-4,-1); 
\draw[yslant=0.5,xslant=-1,ultra thick, black](4,0) rectangle +(-1,-1); 
\draw[yslant=0.5,xslant=-1,ultra thick, white](3,0) rectangle +(1,0); 

\draw[yslant=0.5,pattern={crosshatch dots},pattern color=green] (5,-4) rectangle +(1,3);

\draw[yslant=0.5,xslant=-1,pattern={crosshatch dots},pattern color=green](5,0) rectangle +(-1,-1); 
\draw[yslant=0.5,pattern={horizontal lines},pattern color=blue](4,-2) rectangle +(1,1);    
\draw[yslant=0.5,xslant=-1,pattern={horizontal lines},pattern color=blue](7,1) rectangle +(-4,-1);
\draw[yslant=0.5,xslant=-1,pattern={horizontal lines},pattern color=blue](4,0) rectangle +(-1,-1);
%\draw[yslant=-0.5,pattern={horizontal lines},pattern color=blue](4,2) rectangle +(-3,-1);     
\draw[yslant=-0.5,pattern={horizontal lines},pattern color=blue](4,3) rectangle +(-2,-1);
  \draw[yslant=-0.5] (1,0) grid (4,3);
  \draw[yslant=0.5] (4,-4) grid (8,-1);
  \draw[yslant=0.5,xslant=-1] (3,-1) grid (7,2);
\node[anchor=south west] at (-0.25,1.5,0) {$S_1$};
\node[anchor=south west] at (-0.25,0.5,0) {$S_2$};
\node[anchor=south west] at (-0.25,-0.5,0) {$S_3$};
\node[anchor=north west] at (0.5,3.8,0) {$Q_1$};
\node[anchor=north west] at (1.5,4.3,0) {$Q_2$};
\node[anchor=north west] at (2.5,4.8,0) {$Q_3$};
\node[anchor=north west] at (3.5,5.3,0) {$Q_4$};
\node[anchor=north west] at (0.4,-0.8,0) {$T_3$};
\node[anchor=north west] at (1.4,-1.3,0) {$T_2$};
\node[anchor=north west] at (2.4,-1.8,0) {$T_1$};
\node[anchor=north] at (7,5,0) {$R_3$};
\end{tikzpicture}}
&

\small{
\begin{tikzpicture}[scale=0.425]

  \draw[yslant=-0.5] (1,0) grid (4,3);
  \draw[yslant=0.5] (4,-4) grid (8,-1);
  \draw[yslant=0.5,xslant=-1] (3,-1) grid (7,2);
\node[anchor=south west] at (-0.25,1.5,0) {$S_1$};
\node[anchor=south west] at (-0.25,0.5,0) {$S_2$};
\node[anchor=south west] at (-0.25,-0.5,0) {$S_3$};
\node[anchor=north west] at (0.5,3.8,0) {$Q_1$};
\node[anchor=north west] at (1.5,4.3,0) {$Q_2$};
\node[anchor=north west] at (2.5,4.8,0) {$Q_3$};
\node[anchor=north west] at (3.5,5.3,0) {$Q_4$};
\node[anchor=north west] at (0.4,-0.8,0) {$T_3$};
\node[anchor=north west] at (1.4,-1.3,0) {$T_2$};
\node[anchor=north west] at (2.4,-1.8,0) {$T_1$};
\node[anchor=north] at (7,5,0) {$R_4$};
\end{tikzpicture}}\\
\end{tabular}
    
\end{center}

\caption{The third candidate allocation $A''=(S_1',Q_2^*)$. We use blue for bundle $S_1'  \subseteq S_1$ and green color for $Q_2^*\subseteq Q_2$. The cut $C=\{T_2^* \cup Q_1^*\}$ is shown with a thick line. Note that the allocation is valid and none of the bundles intersects with $R_4$. We try to apply Corollary \ref{cor:2agents} for agents $T$ and $R$ and the set of items $M \setminus (T_2^*\cup Q_1^*)$. We could construct a similar allocation for any bundle $S_i'$ and bundle $Q_j$ by renaming the bundles such that $S_i' \cap Q_j^* = \emptyset$ and $S_i' \cap R_4 = \emptyset$.}
\label{fig:4agents3}
\end{figure}
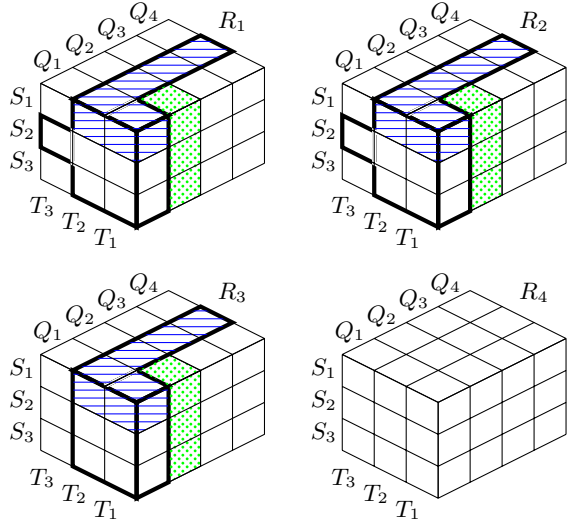

   Otherwise, it should be that for the cut $C=S'_1\cup Q^*_2$ there exist two sets $T_1',T_2' \in \mathcal{X}_{T}(C)$. 
    Since for any $j \in \{1,2\}$, $T_j'\subseteq S'_1\cup Q^*_2$, by \eqref{eq:4agentsCond2} and \eqref{eq:4agentsCond3} it holds that 
      \begin{equation}
        \label{eq:4agentsCond4}
        T_j' \cap Q^*_3 = \emptyset, T_j' \cap S'_2 = \emptyset \mbox{ and } T_j' \cap R_4 = \emptyset, \mbox{ for } j \in \{1,2\}\,.
    \end{equation}

\begin{figure}[t]
\begin{center}

    \begin{tabular}{cc}
\small{
\begin{tikzpicture}[scale=0.425]
\draw[yslant=0.5,ultra thick, black] (5,-4) rectangle +(1,3);
\draw[yslant=0.5,ultra thick, black](4,-2) rectangle +(1,1); \draw[yslant=0.5,xslant=-1,ultra thick, black](5,1) rectangle +(-2,-2);   

\draw[yslant=0.5,xslant=-1,ultra thick, black](7,1) rectangle +(-2,-1); 
\draw[yslant=-0.5,ultra thick, black](4,3) rectangle +(-2,-1);

\draw[yslant=0.5,xslant=-1,ultra thick, white](5,1) rectangle +(0,-1); 

\draw[yslant=0.5,ultra thick, white] (5,-2) rectangle +(0,1);

\draw[yslant=0.5,pattern={crosshatch dots},pattern color=green] (6,-4) rectangle +(1,3);
\draw[yslant=0.5,pattern={north west lines},pattern color=red] (5,-4) rectangle +(1,3);
\draw[yslant=0.5,pattern={north west lines},pattern color=red] (4,-2) rectangle +(2,1);
\draw[yslant=0.5,xslant=-1,pattern={crosshatch dots},pattern color=green](6,0) rectangle +(-1,-1); 
\draw[yslant=0.5,pattern={horizontal lines},pattern color=blue](4,-3) rectangle +(1,1);    
%\draw[yslant=0.5,xslant=-1,pattern={north west lines},pattern color=red](7,1) rectangle +(-4,-1);
\draw[yslant=0.5,xslant=-1,pattern={north west lines},pattern color=red](5,0) rectangle +(-2,-1);
\draw[yslant=-0.5,pattern={horizontal lines},pattern color=blue](4,2) rectangle +(-3,-1);     
\draw[yslant=-0.5,pattern={north west lines},pattern color=red](4,3) rectangle +(-1,-1);    
  \draw[yslant=-0.5] (1,0) grid (4,3);
  \draw[yslant=0.5] (4,-4) grid (8,-1);
  \draw[yslant=0.5,xslant=-1] (3,-1) grid (7,2);
\node[anchor=south west] at (-0.25,1.5,0) {$S_1$};
\node[anchor=south west] at (-0.25,0.5,0) {$S_2$};
\node[anchor=south west] at (-0.25,-0.5,0) {$S_3$};
\node[anchor=north west] at (0.5,3.8,0) {$Q_1$};
\node[anchor=north west] at (1.5,4.3,0) {$Q_2$};
\node[anchor=north west] at (2.5,4.8,0) {$Q_3$};
\node[anchor=north west] at (3.5,5.3,0) {$Q_4$};
\node[anchor=north west] at (0.4,-0.8,0) {$T_3$};
\node[anchor=north west] at (1.4,-1.3,0) {$T_2$};
\node[anchor=north west] at (2.4,-1.8,0) {$T_1$};
\node[anchor=north] at (7,5,0) {$R_1$};
\end{tikzpicture}}
&
\small{
\begin{tikzpicture}[scale=0.425]
\draw[yslant=0.5,ultra thick, black] (5,-4) rectangle +(1,3);
\draw[yslant=0.5,ultra thick, black](4,-2) rectangle +(1,1); \draw[yslant=0.5,xslant=-1,ultra thick, black](5,1) rectangle +(-2,-2);   

\draw[yslant=0.5,xslant=-1,ultra thick, black](7,1) rectangle +(-2,-1); 
\draw[yslant=-0.5,ultra thick, black](4,3) rectangle +(-2,-1);

\draw[yslant=0.5,xslant=-1,ultra thick, white](5,1) rectangle +(0,-1); 

\draw[yslant=0.5,ultra thick, white] (5,-2) rectangle +(0,1);
\draw[yslant=0.5,pattern={crosshatch dots},pattern color=green] (6,-4) rectangle +(1,3);
\draw[yslant=0.5,pattern={north west lines},pattern color=red] (5,-4) rectangle +(1,3);
\draw[yslant=0.5,pattern={north west lines},pattern color=red] (4,-2) rectangle +(2,1);
\draw[yslant=0.5,xslant=-1,pattern={crosshatch dots},pattern color=green](6,0) rectangle +(-1,-1); 
\draw[yslant=0.5,pattern={horizontal lines},pattern color=blue](4,-3) rectangle +(1,1);    
%\draw[yslant=0.5,xslant=-1,pattern={north west lines},pattern color=red](7,1) rectangle +(-4,-1);
\draw[yslant=0.5,xslant=-1,pattern={north west lines},pattern color=red](5,0) rectangle +(-2,-1);
\draw[yslant=-0.5,pattern={horizontal lines},pattern color=blue](4,2) rectangle +(-3,-1);     
\draw[yslant=-0.5,pattern={north west lines},pattern color=red](4,3) rectangle +(-1,-1);    
  \draw[yslant=-0.5] (1,0) grid (4,3);
  \draw[yslant=0.5] (4,-4) grid (8,-1);
  \draw[yslant=0.5,xslant=-1] (3,-1) grid (7,2);
\node[anchor=south west] at (-0.25,1.5,0) {$S_1$};
\node[anchor=south west] at (-0.25,0.5,0) {$S_2$};
\node[anchor=south west] at (-0.25,-0.5,0) {$S_3$};
\node[anchor=north west] at (0.5,3.8,0) {$Q_1$};
\node[anchor=north west] at (1.5,4.3,0) {$Q_2$};
\node[anchor=north west] at (2.5,4.8,0) {$Q_3$};
\node[anchor=north west] at (3.5,5.3,0) {$Q_4$};
\node[anchor=north west] at (0.4,-0.8,0) {$T_3$};
\node[anchor=north west] at (1.4,-1.3,0) {$T_2$};
\node[anchor=north west] at (2.4,-1.8,0) {$T_1$};
\node[anchor=north] at (7,5,0) {$R_2$};
\end{tikzpicture}}\\
\small{
\begin{tikzpicture}[scale=0.425]
\draw[yslant=0.5,ultra thick, black] (5,-4) rectangle +(1,3);
\draw[yslant=0.5,ultra thick, black](4,-2) rectangle +(1,1); \draw[yslant=0.5,xslant=-1,ultra thick, black](5,1) rectangle +(-2,-2);   

\draw[yslant=0.5,xslant=-1,ultra thick, black](7,1) rectangle +(-2,-1); 
\draw[yslant=-0.5,ultra thick, black](4,3) rectangle +(-2,-1);

\draw[yslant=0.5,xslant=-1,ultra thick, white](5,1) rectangle +(0,-1); 

\draw[yslant=0.5,ultra thick, white] (5,-2) rectangle +(0,1);
\draw[yslant=0.5,pattern={crosshatch dots},pattern color=green] (6,-4) rectangle +(1,3);
\draw[yslant=0.5,pattern={north west lines},pattern color=red] (5,-4) rectangle +(1,3);
\draw[yslant=0.5,pattern={north west lines},pattern color=red] (4,-2) rectangle +(2,1);
\draw[yslant=0.5,xslant=-1,pattern={crosshatch dots},pattern color=green](6,0) rectangle +(-1,-1); 
\draw[yslant=0.5,pattern={horizontal lines},pattern color=blue](4,-3) rectangle +(1,1);    
%\draw[yslant=0.5,xslant=-1,pattern={north west lines},pattern color=red](7,1) rectangle +(-4,-1);
\draw[yslant=0.5,xslant=-1,pattern={north west lines},pattern color=red](5,0) rectangle +(-2,-1);
\draw[yslant=-0.5,pattern={horizontal lines},pattern color=blue](4,2) rectangle +(-2,-1);     
\draw[yslant=-0.5,pattern={north west lines},pattern color=red](4,3) rectangle +(-1,-1);    
  \draw[yslant=-0.5] (1,0) grid (4,3);
  \draw[yslant=0.5] (4,-4) grid (8,-1);
  \draw[yslant=0.5,xslant=-1] (3,-1) grid (7,2);
\node[anchor=south west] at (-0.25,1.5,0) {$S_1$};
\node[anchor=south west] at (-0.25,0.5,0) {$S_2$};
\node[anchor=south west] at (-0.25,-0.5,0) {$S_3$};
\node[anchor=north west] at (0.5,3.8,0) {$Q_1$};
\node[anchor=north west] at (1.5,4.3,0) {$Q_2$};
\node[anchor=north west] at (2.5,4.8,0) {$Q_3$};
\node[anchor=north west] at (3.5,5.3,0) {$Q_4$};
\node[anchor=north west] at (0.4,-0.8,0) {$T_3$};
\node[anchor=north west] at (1.4,-1.3,0) {$T_2$};
\node[anchor=north west] at (2.4,-1.8,0) {$T_1$};
\node[anchor=north] at (7,5,0) {$R_3$};
\end{tikzpicture}}
&

\small{
\begin{tikzpicture}[scale=0.425]
\draw[yslant=0.5,pattern={vertical lines},pattern color=magenta] (4,-4) rectangle +(4,3);
\draw[yslant=-0.5,pattern={vertical lines},pattern color=magenta](4,3) rectangle +(-3,-3);
\draw[yslant=0.5,xslant=-1,pattern={vertical lines},pattern color=magenta] (7,2) rectangle +(-4,-3);
    
  \draw[yslant=-0.5] (1,0) grid (4,3);
  \draw[yslant=0.5] (4,-4) grid (8,-1);
  \draw[yslant=0.5,xslant=-1] (3,-1) grid (7,2);
\node[anchor=south west] at (-0.25,1.5,0) {$S_1$};
\node[anchor=south west] at (-0.25,0.5,0) {$S_2$};
\node[anchor=south west] at (-0.25,-0.5,0) {$S_3$};
\node[anchor=north west] at (0.5,3.8,0) {$Q_1$};
\node[anchor=north west] at (1.5,4.3,0) {$Q_2$};
\node[anchor=north west] at (2.5,4.8,0) {$Q_3$};
\node[anchor=north west] at (3.5,5.3,0) {$Q_4$};
\node[anchor=north west] at (0.4,-0.8,0) {$T_3$};
\node[anchor=north west] at (1.4,-1.3,0) {$T_2$};
\node[anchor=north west] at (2.4,-1.8,0) {$T_1$};
\node[anchor=north] at (7,5,0) {$R_4$};
\end{tikzpicture}}\\
\end{tabular}
    
\end{center}
\caption{The final allocation $A^*=(S_2',T_1',Q_3^*,R_4)$. We use blue for bundle $S_2' \subseteq S_2$, red for bundle $T_1' \subseteq T_1$, and green for $Q_3^* \subseteq Q_3$. We use magenta to denote bundle $R_4$. The cut is shown with a thick line $C=\{S_1',Q_3^*\}$ The allocation is valid.}
\label{fig:4agents4}
\end{figure}
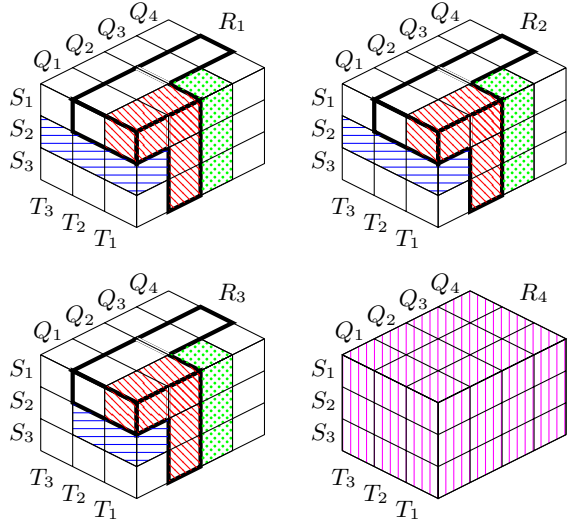

    {\bf Final allocation.} Finally, we offer the allocation $A^*=(S_2',T_1',Q_3^*,R_4)$ which is valid and each agent has value at least $1/2$ for their allocated bundle. Hence, the lemma follows. Figure \ref{fig:4agents4} illustrates the allocation.
\end{proof}

As a corollary of Lemma~\ref{Lemma:4agents}, a 1/2-MMS allocation always
exists for four agents (by Observation 1); this bound is also
tight \cite{GhodsiHSSY22}. 
\begin{corollary}
\label{cor:4agents}
    A $1/2$-MMS allocation exists for four agents with subadditive valuation functions. 
\end{corollary}

\subsection{Many agents}
\label{sec:MoreAgents}
In this section, we demonstrate how our arguments developed in the previous sections can be useful towards proving positive results for the case of multiple agents. Indeed, we show the existence of $1/2$-MMS allocations for multiple agents, when they have one of two admissible valuation functions.

\begin{theorem}\label{thm:two-types}
    For every instance of $n$ agents, where each agent $i$ has a valuation function  $v_i \in \{v_S,v_T\}$, for any subadditive valuation functions $v_S,v_T$, there exists a $1/2$-MMS allocation.
\end{theorem}
\begin{proof}

    The proof is by induction on the number of agents. At the induction step we guarantee that $\mu_i$ for each remaining agent $i$ does not decrease, so at the end they will receive at least $1/2$ of their original $\mu_i$ value. For $n=2$ the theorem follows by Corollary \ref{cor:2agents}. 
    Let's assume that the statement holds for less than $n$, we will show that it also works for $n$ agents. Let $n_S$ and $n_T$ be the number of agents with valuation function $v_S$ (agents of type $S$), and $v_T$ (agents of type $T$). Note that $n_S + n_T = n$.  W. l. o. g. assume that $n_S\ge n_T$ and hence $n_S \ge \left\lceil \frac{n}{2}\right\rceil$ and $n_T \le \left\lfloor \frac{n}{2}\right\rfloor$. Let also $S_j$ and $T_j$ be the $j$-th MMS bundle of an agent of type $S$ and $T$, respectively.
    
    We consider the cut $C=\bigcup_{j=1}^{\lfloor n/2\rfloor}T_j$, i.e., the union of the first $\lfloor n/2\rfloor$ MMS bundles of the agents of type $T$; note that both $C$ and $M\setminus C$ contain at least $\lfloor n/2\rfloor$ such MMS bundles.
    Then, we consider the maximum desired half, $\mathcal{X}_{S}^{*}(C)$, of agents of type $S$ over $C$. 
    Let $n'=\min\{\lvert \mathcal{X}_{S}^{*}(C) \rvert,n_S\}$, and by Observation~\ref{obs:Cut_noOfDesiredSets}, $n' \ge \left\lceil\frac{n}{2}\right\rceil$. This implies that there exist $n'$ mutually disjoint bundles, each of which has value at least $1/2$ for the agents of type $S$. Suppose that we assign those bundles to $n'\leq n_S$ agents of type $S$. 
    
    Let $M'$ be the union of those bundles, then it holds that $M'$ is disjoint with either $C$ or $M\setminus C$,  therefore $M'$ is disjoint with at least $\lfloor n/2\rfloor \geq n-n'$ MMS bundles of agents of type $T$. Moreover $M'$ is a subset of $n'$ MMS bundles of agents of type $S$, therefore $M'$ is disjoint with $n-n'$ MMS bundles of agents of type $S$. Altogether,  we are left with a reduced instance with $n-n'$ agents, where each remaining agent $i$  can partition the remaining items $M\setminus M'$ into at least $n-n'$ bundles of value at least $\mu_i^n(M)$, since $\mu_i^{n-n'}(M\setminus M')\geq \mu_i^n(M)$.  By the induction hypothesis there exists a $1/2$-MMS allocation for the reduced instance, and by combining it with the allocation of $M'$ to the $n'$ agents the proof is complete.
\end{proof}

 \section{$\boldsymbol{\alpha}$-MMS($\mathbf{d}$) for subadditive valuations}
\label{sec:Reductions}
In this section, we present a thorough study of conditions of existence (and non-existence) of $\boldsymbol{\alpha}$-MMS($\mathbf{d}$) allocations for various combinations of $\boldsymbol{\alpha}$ and $\mathbf{d}$. We provide two characterization results for three agents and several impossibility results for many agents. 

\subsection{Three Agents}
 In this section we consider three agents with subadditive valuations, and we fully characterize the conditions of existence of $1/2$-MMS$(\mathbf{d})$  (Theorem~\ref{thm:three-general}) and of $(1,1/2,1/2)$-MMS$(\mathbf{d})$ allocations (Theorem~\ref{thm:three-general-approximate}), with respect to any vector $\mathbf{d}$. We prove those results via a series of Lemmas and Corollaries in Sections~\ref{sec:CharacterizationLemmas}~and~\ref{sec:ManyAgentsImpossib}.

 We give a key claim (Claim~\ref{cl:disjointAllocations}) to be used as a technical tool in the follow up lemmas; if there exist two allocations that ``satisfy'' all but one agent, and those allocations are disjoint in one MMS bundle of the last agent, then one of the two allocations will leave items that the last agent values by at least $1/2$. Hence, that allocation can be extended to include the last agent that is guaranteed to receive at least $1/2$. 

 \begin{claim}
\label{cl:disjointAllocations}
    Consider any instance of $n$ agents, and any vectors ${\bf d}=(d_1,\ldots, d_n)$ and $\boldsymbol{\alpha}=(\alpha_1,\ldots,\alpha_n)$ with $\alpha_n =1/2$. If there exist two  allocations $A$ and $A'$, that are both $\boldsymbol{\alpha}_{-n}$-MMS$({\bf d}_{-n})$ for the first $n-1$ agents, such that there exists an MMS bundle $X$ of the last agent, where $X \cap \left(\bigcup_{Y\in A}Y\right)$ and $X \cap \left(\bigcup_{Y'\in A'}Y'\right)$ are disjoint,  
    then there exists a $\boldsymbol{\alpha}$-MMS$(\bf d)$ allocation for all $n$ agents.
\end{claim}

\begin{proof}
    The last agent values by at least $1/2$ either $X\cap \left(\bigcup_{Y\in A}Y\right)$ or $X\setminus \left(\bigcup_{Y\in A}Y\right) \subseteq X\cap \left(\bigcup_{Y'\in A'}Y'\right)$. W.l.o.g. suppose that the former holds. Then, the allocation where the last agent gets $X\cap \left(\bigcup_{Y\in A}Y\right) $ and the others get their allocated bundle in $A'$ is valid and it is $\boldsymbol{\alpha}$-MMS$(\bf d)$.
\end{proof}

\subsubsection{Characterizations of $1/2$-MMS($\mathbf{d}$) guarantees}
In this section we provide a complete characterization of results regarding $1/2$-MMS$(\mathbf{d})$ for three agents with subadditive valuations, for any vector $\mathbf{d}=(d_1,d_2,d_3)$. We summarize the results in the following theorem; note that we use the value $\sum_{i=1}^3d_i$ to distinguish among different $\mathbf{d}$, however, we do not claim that there is any strong correlation. 

\begin{theorem}\label{thm:three-general}
    A $1/2$-MMS$(\mathbf{d})$ allocation exists for three agents with subadditive valuation functions, when (i) $\mathbf{d}=(3,2,2)$ or (ii)  $\sum_{i=1}^3(d_i)\geq 8$ and $d_i=1$ for at most one agent $i$. In any other case, there exists an instance with no $1/2$-MMS$(\mathbf{d})$ allocation.
\end{theorem}
\begin{proof}
    The positive results are derived by using Observation~\ref{obs:simpleReduction} and showing that there is always a $1/2$-MMS$(3,2,2)$ allocation (Lemma~\ref{Lemma:threehalfs}), and a $(1,1/2,1/2)$-MMS$(\mathbf{d})$ allocation, for $\mathbf{d}=(5,2,1)$ (Lemma~\ref{lem:(5,2,1)}), $\mathbf{d}=(4,3,1)$ (Lemma~\ref{lem:(4,3,1)}), and $\mathbf{d}=(4,2,2)$  (Lemma~\ref{lemma:3Agents4}). 
    
    The impossibility results are derived by using Observation~\ref{obs:simpleReduction} and showing that for any of the following $\mathbf{d}$, there exists an instance that no $1/2$-MMS$(\mathbf{d})$ allocation exists. This is shown for $\mathbf{d}=(k,1,1)$, for any $k\geq 1$ (Corollary~\ref{cor:d_i<k}), $\mathbf{d}=(4,2,1)$ (Lemma~\ref{lem:imposs(4,2,1)}), $\mathbf{d}=(3,3,1)$ (Lemma~\ref{lem:imposs(3,3,1)}), and $\mathbf{d}=(2,2,2)$ (Lemma~\ref{lem:nagentsImp1}). 
    
    All Lemmas and Corollaries appear in Sections \ref{sec:CharacterizationLemmas}~and~\ref{sec:ManyAgentsImpossib}.
\end{proof}

\subsubsection{Characterizations of $(1,1/2,1/2)$-MMS$(\bf d)$ guarantees }
In this section we provide a complete characterization of results regarding $(1,1/2,1/2)$-MMS$(\bf d)$ for three agents with subadditive valuations.

\begin{theorem}\label{thm:three-general-approximate}
    A $(1,1/2,1/2)$-MMS$(\bf d)$ allocation exists for three agents with subadditive valuation functions, when $\sum_{i=1}^3(d_i)\geq 8$, when $d_i=1$ for at most one agent $i$, and $\max_i d_i \geq 4$. In any other case, there exists an instance with no $1/2$-MMS$(\mathbf{d})$ allocation.
\end{theorem}

\begin{proof}
    The positive results are derived by using Observation~\ref{obs:simpleReduction} and showing that there is always a  $(1,1/2,1/2)$-MMS$(\mathbf{d})$ allocation, for $\mathbf{d}=(5,2,1)$ (Lemma~\ref{lem:(5,2,1)}), $\mathbf{d}=(4,3,1)$ (Lemma~\ref{lem:(4,3,1)}) and $\mathbf{d}=(4,2,2)$ (Lemma~\ref{lemma:3Agents4}). 
    
    The impossibility results are derived by using Observation~\ref{obs:simpleReduction} and showing that for any of the following $\mathbf{d}$, there exists an instance that no $1/2$-MMS$(\mathbf{d})$ allocation exists. This is shown for $\mathbf{d}=(k,1,1)$, for any $k\geq 1$ (Corollary~\ref{cor:d_i<k}), $\mathbf{d}=(4,2,1)$ (Lemma~\ref{lem:imposs(4,2,1)}), and $\mathbf{d}=(3,3,3)$ (Lemma~\ref{lem:imposs(3,3,3)}).

    All Lemmas and Corollaries appear in Sections \ref{sec:CharacterizationLemmas}~and~\ref{sec:ManyAgentsImpossib}.
    \end{proof}

\subsubsection{$\alpha$-MMS$(\bf d)$ guarantees for three agents}
\label{sec:CharacterizationLemmas}

In this section we give a series of lemmas regarding the existence of $1/2$-MMS$(\bf d)$ and $(1,1/2,1/2)$-MMS$(\bf d)$ allocations, or their impossibilities, for three agents. Those lemmas are used in the characterizations of the existence of $1/2$-MMS$(\bf d)$ and $(1,1/2,1/2)$-MMS$(\bf d)$ allocations of the previous sections, focusing on results for three agents.

    \begin{lemma}
\label{lem:imposs(4,2,1)}
There exists an instance of three agents with subadditive valuation functions, where no $1/2$-MMS$(4,2,1)$ allocation exists.
\end{lemma}

\begin{proof}
    Consider an instance with $M=\{h,g_1,g_2,g_3\}$, where $v_1(S)=1$ for any $S\subseteq M$, $v_2(S)=1$, if $h\in S$, otherwise $v_2(S)=1/3 \cdot |S|$, and the last agent has an additive valuation function over $M$, with $v_3(h)=1/3$ and $v_3(g)=2/9$ for any $g\neq h$. Note that the first agent can partition $M$ into four bundles, each valued at $1$, the second agent can partition $M$ into $(\{h\},\{g_1,g_2,g_3\})$, valuing each bundle at $1$, and the third agent values $M$ at $1$.

    Suppose that there exists an allocation $A$ that is $1/2$-MMS$(4,2,1)$. If $h\in A_3$, then the third agent should receive at least one more item. Then, the second agent should receive two items other than $h$ to form a bundle with a value of at least $1/2$ to her; this leaves no items for the first agent. If $h\notin A_3$, the third agent must receive all remaining items to form a bundle with a value of at least $1/2$. This leaves one item for each of the first two agents. We conclude that no such allocation $A$ exists.  
\end{proof}

\begin{lemma} 
\label{lem:(5,2,1)}
A $(1,1/2,1/2)$-MMS$(5,2,1)$ allocation exists for three agents with subadditive valuation functions.
\end{lemma}
\begin{proof}
    We denote by $S,T,Q$ the three agents, and by $S_j,T_j,Q_j$ their $j$-th MMS bundle, respectively. 
    Consider all the cuts $C_{ijk}=\{S_i \cup S_j \cup S_k\}$ that we offer to $T$. We first show the following claim:
    \begin{claim}
        There exists a cut $C_{ijk}$, such that $v_T(T_1\cap C_{ijk})\geq 1/2$ and $v_T(T_2\cap C_{ijk})\geq 1/2$. 
    \end{claim}
    \begin{proof}
        Suppose on the contrary that there is no such cut. Then, for each $C_{ijk}$, agent $T$ has value at least $1/2$ for the intersection of $C_{ijk}$ and either $T_1$ or $T_2$ (but not both due to our assumption). The reason is that if $T$'s value was less than $1/2$ for both intersections, due to subadditivity any $C_{i'j'k'}$ for which $C_{ijk} \cup C_{i'j'k'}=M$ would contradict our assumption. 
        Let $\phi(C_{ijk})=1$ if $T$ has value at least $1/2$ for the intersection with $T_1$, otherwise, $\phi(C_{ijk})=2$.

        W.l.o.g. suppose that  $\phi(C_{123})=1$, then $\phi(C_{145})=2$, o.w. $v_T(C_{123}\cap T_2)+v_2(C_{145}\cap T_2)<1$, which is a contradiction to subadditivity, since the union of those two sets is $T_2$, and $v(T_2)\geq 1$. Similarly $\phi(C_{245})=2$ and $\phi(C_{345})=2$. 
        For the same reason, since $\phi(C_{145})=2$, it should be that $\phi(C_{234})=1$, and then $\phi(C_{125})=2$. Finally it should be $\phi(C_{345})=1$, but we argue above that $\phi(C_{345})=2$, which is a contradiction to our assumption. 
    \end{proof}

    W.l.o.g. suppose that $C_{123}$ is the cut satisfying the statement of the above claim. Then, $T$ has value at least $1/2$ for both $T_1^*=T_1\cap C_{123}$ and $T_2^*=T_2\cap C_{123}$ (Figure \ref{fig:3agents521}(a)). Then, the allocations $A=(S_4,T_1^*)$ and  $A'=(S_5,T_2^*)$ for agents $S$ and $T$ are $(1,1/2)$-MMS$(5,2)$ and they are disjoint (Figures \ref{fig:3agents521}(b) and \ref{fig:3agents521}(c) resp.). By Claim~\ref{cl:disjointAllocations} the lemma follows. 
\end{proof}

\begin{center}
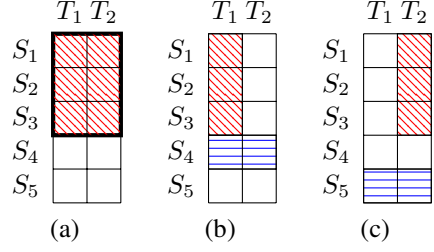
\begin{figure}[t] 
\begin{center}
    \begin{tabular}{ccc}
\begin{tikzpicture}[scale=0.45]
\draw[black, ultra thick](0,2) rectangle +(2,3);
\draw[pattern={north west lines},pattern color=red]
(0,2) rectangle +(2,3);
  \draw[step=1cm,] (0,0) grid (2,5);

\node[anchor=north] at (0.5,6.25) {$T_1$};
\node[anchor=north] at (1.5,6.25) {$T_2$};
\node[anchor=west] at (-1.5,0.5) {$S_5$};
\node[anchor=west] at (-1.5,1.5) {$S_4$};
\node[anchor=west] at (-1.5,2.5) {$S_3$};
\node[anchor=west] at (-1.5,3.5) {$S_2$};
\node[anchor=west] at (-1.5,4.5) {$S_1$};
\end{tikzpicture} & \begin{tikzpicture}[scale=0.45]
%\draw[black, ultra thick](0,2) rectangle +(2,3);
\draw[pattern={north west lines},pattern color=red]
(0,2) rectangle +(1,3);
\draw[pattern={horizontal lines},pattern color=blue](0,1) rectangle +(2,1);
  \draw[step=1cm,] (0,0) grid (2,5);

\node[anchor=north] at (0.5,6.25) {$T_1$};
\node[anchor=north] at (1.5,6.25) {$T_2$};
\node[anchor=west] at (-1.5,0.5) {$S_5$};
\node[anchor=west] at (-1.5,1.5) {$S_4$};
\node[anchor=west] at (-1.5,2.5) {$S_3$};
\node[anchor=west] at (-1.5,3.5) {$S_2$};
\node[anchor=west] at (-1.5,4.5) {$S_1$};
\end{tikzpicture} &
\begin{tikzpicture}[scale=0.45]
%\draw[black, ultra thick](0,2) rectangle +(2,3);
\draw[pattern={north west lines},pattern color=red]
(1,2) rectangle +(1,3);
\draw[pattern={horizontal lines},pattern color=blue](0,0) rectangle +(2,1);
  \draw[step=1cm,] (0,0) grid (2,5);

\node[anchor=north] at (0.5,6.25) {$T_1$};
\node[anchor=north] at (1.5,6.25) {$T_2$};
\node[anchor=west] at (-1.5,0.5) {$S_5$};
\node[anchor=west] at (-1.5,1.5) {$S_4$};
\node[anchor=west] at (-1.5,2.5) {$S_3$};
\node[anchor=west] at (-1.5,3.5) {$S_2$};
\node[anchor=west] at (-1.5,4.5) {$S_1$};
\end{tikzpicture}\\
(a) & (b) & (c) \\
    \end{tabular}
    \end{center}
\caption{In (a), with a thick line we show the cut $C=C_{123}$, and the red bundles correspond to $T_1^*,T_2^*$. In (b) and (c), we illustrate the allocations $A$ and $A'$, respectively; agent $S$ gets a whole bundle in each allocation (denoted with blue color).}
\label{fig:3agents521}
\end{figure}  
\end{center}

\begin{lemma} 
\label{lem:(4,3,1)}
A $(1,1/2,1/2)$-MMS$(4,3,1)$ allocation exists for three agents with subadditive valuation functions.  
\end{lemma}
\begin{center}
\begin{figure} 
\begin{center}
    \begin{tabular}{ccc}
\begin{tikzpicture}[scale=0.45]
%(0,0) rectangle +(1,2);
%\draw[pattern={north west lines},pattern color=red]
(0,0) rectangle +(1,2);
\draw[black, ultra thick](0,2) rectangle +(3,2);
\draw[pattern={north west lines},pattern color=red]
(0,2) rectangle +(2,2);
  \draw[step=1cm,] (0,0) grid (3,4);

\node[anchor=north] at (0.5,5.25) {$T_1$};
\node[anchor=north] at (1.5,5.25) {$T_2$};
\node[anchor=north] at (2.5,5.25) {$T_3$};
\node[anchor=west] at (-1.5,0.5) {$S_4$};
\node[anchor=west] at (-1.5,1.5) {$S_3$};
\node[anchor=west] at (-1.5,2.5) {$S_2$};
\node[anchor=west] at (-1.5,3.5) {$S_1$};
\end{tikzpicture} & \begin{tikzpicture}[scale=0.45]
%(0,0) rectangle +(1,2);
\draw[pattern={north west lines},pattern color=red]
(0,2) rectangle +(1,2);
\draw[pattern={horizontal lines},pattern color=blue](0,1) rectangle +(3,1);
  \draw[step=1cm,] (0,0) grid (3,4);

\node[anchor=north] at (0.5,5.25) {$T_1$};
\node[anchor=north] at (1.5,5.25) {$T_2$};
\node[anchor=north] at (2.5,5.25) {$T_3$};
\node[anchor=west] at (-1.5,0.5) {$S_4$};
\node[anchor=west] at (-1.5,1.5) {$S_3$};
\node[anchor=west] at (-1.5,2.5) {$S_2$};
\node[anchor=west] at (-1.5,3.5) {$S_1$};
\end{tikzpicture} &
\begin{tikzpicture}[scale=0.45]
\draw[pattern={north west lines},pattern color=red]
(1,2) rectangle +(1,2);
\draw[pattern={horizontal lines},pattern color=blue](0,0) rectangle +(3,1);
  \draw[step=1cm,] (0,0) grid (3,4);

\node[anchor=north] at (0.5,5.25) {$T_1$};
\node[anchor=north] at (1.5,5.25) {$T_2$};
\node[anchor=north] at (2.5,5.25) {$T_3$};
\node[anchor=west] at (-1.5,0.5) {$S_4$};
\node[anchor=west] at (-1.5,1.5) {$S_3$};
\node[anchor=west] at (-1.5,2.5) {$S_2$};
\node[anchor=west] at (-1.5,3.5) {$S_1$};
\end{tikzpicture}\\
(a) & (b) & (c) \\
    \end{tabular}
    \end{center}
\caption{In (a), with a thick line we show the cut $C=\{S_1 \cup S_2\}$. Red bundles correspond to $T_1^*,T_2^*$ (a). In (b) and (c), we illustrate allocations $A$ and $A'$, respectively; agent $S$ gets a whole bundle in each allocation (denoted with blue color).}
\label{fig:3agents431}
\end{figure}
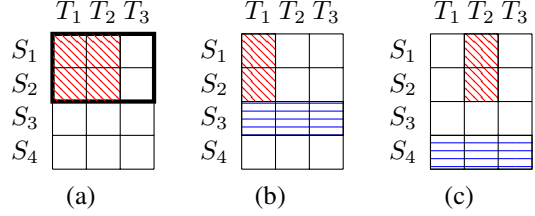  
\end{center}

\begin{proof}
    We denote by $S,T,Q$ the three agents, and by $S_j,T_j,Q_j$ their $j$-th MMS bundle, respectively. Consider the cut $C={S_1\cup S_2}$ that we offer to $T$. W.l.o.g., suppose that
    $\mathcal{X}^*_T(C)=\mathcal{X}_T(C)$ and let $T_1^*,T_2^* \in \mathcal{X}^*_T(C)$ (Figure \ref{fig:3agents431} (a)). 
    Then, the allocations $A=(S_3,T_1^*)$ and  $A'=(S_4,T_2^*)$ (Figures \ref{fig:3agents431} (b) and \ref{fig:3agents431}(c) resp.) for agents $S$ and $T$ are $(1,1/2)$-MMS$(4,3)$ and they are disjoint. By Claim~\ref{cl:disjointAllocations} the lemma follows.
\end{proof}

\begin{lemma} \label{lemma:3Agents4}
A $(1,1/2,1/2)$-MMS$(4,2,2)$ allocation exists for three agents with subadditive valuation functions. 
\end{lemma}
\begin{proof}
We denote by $S,T,Q$ the three agents, and by $S_j,T_j,Q_j$ their $j$-th MMS bundle, respectively. We show the existence of two allocations for agents $T$ and $S$ that are disjoint on some $Q_j$,   and by Claim \ref{cl:disjointAllocations} the proof completes. 

We define the cuts $C_{ij}=\{S_i \cup S_j\}$ that we offer to $T$. W.l.o.g. assume that $v_T(T'_1)\geq 1/2$, for $T_1' =T_1\cap  C_{12}$ and $v_T(T_1'')\geq 1/2$ for $T_1''=T_1 \cap C_{23}$; if this is not the case we can rename the bundles of agent $S$.  Figures \ref{fig:3agents422}(a) and \ref{fig:3agents422}(b) illustrate the bundles and the cuts, respectively.

Next we offer the cut $C=\{(Q_1 \setminus S_1) \cup S_3\}$ to $T$ (Figure \ref{fig:3agents422a}(a)). We consider the intersection of $T_2$ with $C$, and let $T_2^*=\max\{T_2\cap C, T_2\cap (M\setminus C)\}$. Note that $M\setminus C = (Q_2 \setminus S_3) \cup S_1$, so the cut provides some symmetry between $S_1$ and $S_3$; $S_1$ intersect with $T'_1$ but not $T''_1$, and for $S_3$ is the other way around. So, no matter which set is $T_2^*$, it does not intersect either $S_1$ or $S_3$, which in turns does not intersect with either $T'_1$ or $T''_1$. Therefore it is w.l.o.g. to assume that $T_2^*=T_2\cap C$ (Figure \ref{fig:3agents422a}(a)) and then, the allocations $A=(S_1,T_2^*)$ and  $A'=(S_4,T_1'')$ (Figures \ref{fig:3agents422a}(b) and \ref{fig:3agents422a}(c) resp.) for agents $S$ and $T$ are $(1,1/2)$-MMS$(4,2)$ and they are disjoint on $Q_2$. By Claim~\ref{cl:disjointAllocations} the lemma follows.
\end{proof}

\begin{center}

\begin{figure}[t]
\begin{center}

\begin{tabular}{cc}

\small{\begin{tikzpicture}[scale=0.45]
    \draw[yslant=0.5,xslant=-1,color=black,ultra thick] (6,3) rectangle +(-2,-2);    
    \draw[yslant=0.5,color=black,ultra thick] (3,-1) rectangle +(2,2);
    \draw[yslant=-0.5,color=black,ultra thick](3,4) rectangle +(-2,-2);
    \draw[yslant=0.5,xslant=-1,pattern={north west lines},pattern color=red](6,3) rectangle +(-2,-1);
    \draw[yslant=-0.5,pattern={north west lines},pattern color=red](2,4) rectangle +(-1,-2);
  \draw[yslant=-0.5] (1,0) grid (3,4);
  \draw[yslant=0.5] (3,-3) grid (5,1);
  \draw[yslant=0.5,xslant=-1] (4,1) grid (6,3);
\node[anchor=south west] at (-0.2,2.5,0) {$S_1$};
\node[anchor=south west] at (-0.2,1.5,0) {$S_2$};
\node[anchor=south west] at (-0.2,0.5,0) {$S_3$};
\node[anchor=south west] at (-0.2,-0.5,0) {$S_4$};
\node[anchor=north west] at (0.45,4.7,0) {$Q_1$};
\node[anchor=north west] at (1.5,5.3,0) {$Q_2$};
\node[anchor=north west] at (0.5,-0.8,0) {$T_1$};
\node[anchor=north west] at (1.5,-1.2,0) {$T_2$};
\end{tikzpicture}} &
\small{\begin{tikzpicture}[scale=0.45]
     %\draw[yslant=0.5,xslant=-1,color=black,ultra thick] (6,3) rectangle +(-2,-2);    
    \draw[yslant=0.5,color=black,ultra thick] (3,-2) rectangle +(2,2);
    \draw[yslant=-0.5,color=black,ultra thick](3,3) rectangle +(-2,-2);
    %\draw[yslant=0.5,xslant=-1,pattern={north west lines},pattern color=red](6,2) rectangle +(-2,-1);
    \draw[yslant=-0.5,pattern={north west lines},pattern color=red](2,3) rectangle +(-1,-2);
    
  \draw[yslant=-0.5] (1,0) grid (3,4);
  \draw[yslant=0.5] (3,-3) grid (5,1);
  \draw[yslant=0.5,xslant=-1] (4,1) grid (6,3);
\node[anchor=south west] at (-0.2,2.5,0) {$S_1$};
\node[anchor=south west] at (-0.2,1.5,0) {$S_2$};
\node[anchor=south west] at (-0.2,0.5,0) {$S_3$};
\node[anchor=south west] at (-0.2,-0.5,0) {$S_4$};
\node[anchor=north west] at (0.45,4.7,0) {$Q_1$};
\node[anchor=north west] at (1.5,5.3,0) {$Q_2$};
\node[anchor=north west] at (0.5,-0.8,0) {$T_1$};
\node[anchor=north west] at (1.5,-1.2,0) {$T_2$};
\end{tikzpicture}} \\
(a) & (b) \\
$T_1' \in \mathcal{X}_T(C_{12})$ &  $T_1'' \in \mathcal{X}_T(C_{23})$ \\
\end{tabular}
    
\end{center}
\caption{In (a) and (b), with a thick line we show the cut $C_{12}$ and $C_{23}$, respectively. Red bundles correspond to $T_1'$ in (a) and $T_1''$ in (b).}
\label{fig:3agents422}
\end{figure}
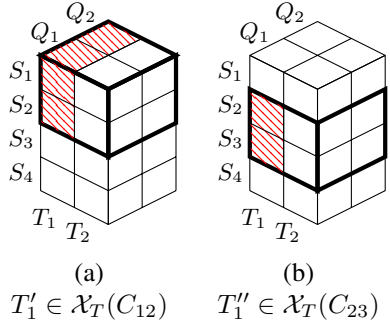 
\end{center}

\begin{center}

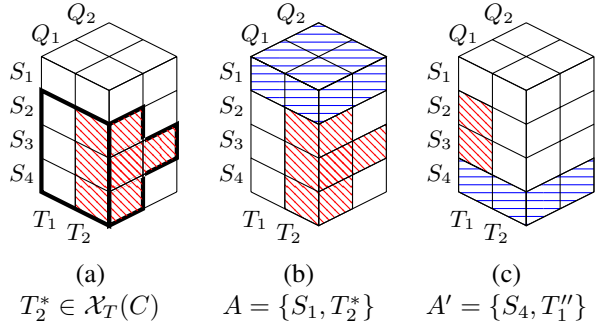
\begin{figure}[t]
\begin{center}

\begin{tabular}{ccc}
\small{\begin{tikzpicture}[scale=0.45]
      
    \draw[yslant=0.5,color=black,ultra thick] (3,-2) rectangle +(2,1);
    \draw[yslant=0.5,color=black,ultra thick] (3,-3) rectangle +(1,3);
    \draw[yslant=0.5,color=black,ultra thick,white] (3,-2) rectangle +(1,1);

    \draw[yslant=0.5,pattern={north west lines},pattern color=red] (3,-3) rectangle +(1,3);
    \draw[yslant=0.5,pattern={north west lines},pattern color=red] (3,-2) rectangle +(2,1);
    \draw[yslant=-0.5,pattern={north west lines},pattern color=red](3,3) rectangle +(-1,-3);
    
    \draw[yslant=-0.5,color=black,ultra thick](3,3) rectangle +(-2,-3);
  \draw[yslant=-0.5] (1,0) grid (3,4);
  \draw[yslant=0.5] (3,-3) grid (5,1);
  \draw[yslant=0.5,xslant=-1] (4,1) grid (6,3);
\node[anchor=south west] at (-0.2,2.5,0) {$S_1$};
\node[anchor=south west] at (-0.2,1.5,0) {$S_2$};
\node[anchor=south west] at (-0.2,0.5,0) {$S_3$};
\node[anchor=south west] at (-0.2,-0.5,0) {$S_4$};
\node[anchor=north west] at (0.45,4.7,0) {$Q_1$};
\node[anchor=north west] at (1.5,5.3,0) {$Q_2$};
\node[anchor=north west] at (0.5,-0.8,0) {$T_1$};
\node[anchor=north west] at (1.5,-1.2,0) {$T_2$};
\end{tikzpicture}}
 &
 \small{\begin{tikzpicture}[scale=0.45]

    \draw[yslant=0.5,pattern={north west lines},pattern color=red] (3,-3) rectangle +(1,3);
    \draw[yslant=0.5,pattern={north west lines},pattern color=red] (3,-2) rectangle +(2,1);
    \draw[yslant=-0.5,pattern={north west lines},pattern color=red](3,3) rectangle +(-1,-3);
    \draw[yslant=0.5,xslant=-1,pattern={horizontal lines},pattern color=blue](3,3) (6,3) rectangle +(-2,-2);  
    \draw[yslant=0.5,pattern={horizontal lines},pattern color=blue](3,0) rectangle +(2,1);     
    \draw[yslant=-0.5,pattern={horizontal lines},pattern color=blue](3,4) rectangle +(-2,-1);  
    
     \draw[yslant=-0.5] (1,0) grid (3,4);
  \draw[yslant=0.5] (3,-3) grid (5,1);
  \draw[yslant=0.5,xslant=-1] (4,1) grid (6,3);
\node[anchor=south west] at (-0.2,2.5,0) {$S_1$};
\node[anchor=south west] at (-0.2,1.5,0) {$S_2$};
\node[anchor=south west] at (-0.2,0.5,0) {$S_3$};
\node[anchor=south west] at (-0.2,-0.5,0) {$S_4$};
\node[anchor=north west] at (0.45,4.7,0) {$Q_1$};
\node[anchor=north west] at (1.5,5.3,0) {$Q_2$};
\node[anchor=north west] at (0.5,-0.8,0) {$T_1$};
\node[anchor=north west] at (1.5,-1.2,0) {$T_2$};
\end{tikzpicture}}
 &
 \small{\begin{tikzpicture}[scale=0.45]
\draw[yslant=-0.5,pattern={north west lines},pattern color=red](2,3) rectangle +(-1,-2);
\draw[yslant=0.5,pattern={horizontal lines},pattern color=blue](3,-3) rectangle +(2,1);     
\draw[yslant=-0.5,pattern={horizontal lines},pattern color=blue](3,1) rectangle +(-2,-1);  
  \draw[yslant=-0.5] (1,0) grid (3,4);
  \draw[yslant=0.5] (3,-3) grid (5,1);
  \draw[yslant=0.5,xslant=-1] (4,1) grid (6,3);
\node[anchor=south west] at (-0.2,2.5,0) {$S_1$};
\node[anchor=south west] at (-0.2,1.5,0) {$S_2$};
\node[anchor=south west] at (-0.2,0.5,0) {$S_3$};
\node[anchor=south west] at (-0.2,-0.5,0) {$S_4$};
\node[anchor=north west] at (0.45,4.7,0) {$Q_1$};
\node[anchor=north west] at (1.5,5.3,0) {$Q_2$};
\node[anchor=north west] at (0.5,-0.8,0) {$T_1$};
\node[anchor=north west] at (1.5,-1.2,0) {$T_2$};
\end{tikzpicture}}
\\
(a) & (b) & (c) \\
$T_2^* \in \mathcal{X}_T(C)$ &  $A = \{S_1,T_2^*\}$ & $A' = \{S_4,T_1''\}$  \\
\end{tabular}
    
\end{center}
\caption{For cut $C=\{(Q_1 \setminus S_1) \cup S_3\}$, let $T_2^* \in \mathcal{X}_T(C)$ (as shown in (a)). Then, the allocations $A=(S_1,T_2^*)$ (b) and  $A'=(S_4,T_1'')$ (c) are disjoint on $Q_2$. We use blue for agent $S$ and red for agent $T$'s bundles.}
\label{fig:3agents422a}
\end{figure} 
\end{center}

\begin{lemma} 
\label{lem:imposs(3,3,3)}
There exists an instance of three agents with subadditive valuation functions, where no $(1,1/2,1/2)$ -MMS allocation exists.
\end{lemma}

\begin{proof}
Consider an instance with items $M=\{g(i,j,k)\mid \mbox{ for } i,j,k \le 3\}$
and the following MMS bundles for the agents:
\begin{itemize}
    \item for agent $S$: $S_i=\{g(i,j,k)\mid \forall j,k\}$ for $i\in \{1,2,3\}$,
    \item for agent $T$: $T_j=\{g(i,j,k)\mid \forall i,k\}$ for $j\in \{1,2,3\}$,
    \item for agent $Q$: $Q_k=\{g(i,j,k)\mid \forall i,j\}$ for $k\in \{1,2,3\}$.
\end{itemize}
In the following, we make the convention that $i+1 = (i\mod 3) +1$, $i+2 = (i+1 \mod 3) +1$, and similarly for the other indices $j$ and $k$.

We will construct the values of all agents symmetrically, such that for each agent $R\in \{S,Q,T\}$, her value for set $B \subseteq M$ will be $v_R(B)=\max_i\{v_R(B \cap R_i)\}$, $v_R(R_i)=1$ for any $i\in \{1,2,3\}$, and $v_R(B)\in\{1/2+\epsilon,1/2-\epsilon\}$ for $B\subset R_i$, for some $i\in\{1,2,3\}$ (where those values will be chosen appropriately so that the valuation functions are subadditive). 
Thus we only need to define the valuation functions only for $B \subset R_i$ for any $i\in \{1,2,3\}$.

For any $R\in \{S,Q,T\}$, $i\in \{1,2,3\}$, and $B \subset R_i$ with $\lvert B\rvert \notin \{4,5\}$, 
\begin{equation}
   \label{eq:422T}
       v_R(B)=\begin{cases}
      1/2 + \epsilon, & \text{if } \lvert B\rvert \geq 6\\
      1/2 - \epsilon, & \text{if } \lvert B\rvert \leq 3
      \end{cases}
   \end{equation}

Consider now the bundles of the form $B^*=R_i\cap((X_j\cap Y_{k})\cup Y_{k+1})$, for any $i,j,k\in \{1,2,3\}$, where $X$ and $Y$ are the other two agents but $R$, with $X\neq Y$. Note that the above sets $B^*$ have cardinality $4$. For any $B \subset R_i$ with $\lvert B\rvert \in \{4,5\}$ we define

\begin{equation}
   \label{eq:422T}
       v_R(B)=\begin{cases}
      1/2 + \epsilon, & \text{if $\lvert B\rvert = 4$, and $B$ is some $B^*$ bundle}\\
      1/2 - \epsilon, & \text{if $\lvert B\rvert = 4$, and $B$ is different from all $B^*$ bundle} \\
      1-v_R(R_i\setminus B), & \text{if $\lvert B\rvert = 5$}
      \end{cases}
   \end{equation}

Subadditivity is trivially guaranteed for those valuation functions. We will only show the following claim to verify that the valuations are valid monotone valuations. 

\begin{center}
\begin{figure}[t]
\begin{center}

\begin{tabular}{ccc}
\small{\begin{tikzpicture}[scale=0.475]

\draw[yslant=0.5,pattern={north west lines},pattern color=red] (4,-3) rectangle +(1,2);
\draw[yslant=0.5,pattern={north west lines},pattern color=red] (4,-3) rectangle +(3,1);
\draw[yslant=0.5,xslant=-1,pattern={north west lines},pattern color=red](3,-1) rectangle +(1,1);
\draw[yslant=-0.5,pattern={north west lines},pattern color=red](4,3) rectangle +(-1,-2);
    
\draw[yslant=0.5,pattern={horizontal lines},pattern color=blue](4,-4) rectangle +(3,1);     
\draw[yslant=-0.5,pattern={horizontal lines},pattern color=blue](4,1) rectangle +(-3,-1);

\draw[yslant=-0.5] (1,0) grid (4,3);
\draw[yslant=0.5] (4,-4) grid (7,-1);
\draw[yslant=0.5,xslant=-1] (3,-1) grid (6,2);
\node[anchor=south west] at (-0.25,1.5,0) {$S_1$};
\node[anchor=south west] at (-0.25,0.5,0) {$S_2$};
\node[anchor=south west] at (-0.25,-0.5,0) {$S_3$};
\node[anchor=north west] at (0.3,-0.6,0) {$T_3$};
\node[anchor=north west] at (1.3,-1.1,0) {$T_2$};
\node[anchor=north west] at (2.3,-1.6,0) {$T_1$};
\node[anchor=north west] at (0.4,3.7,0) {$Q_1$};
\node[anchor=north west] at (1.4,4.2,0) {$Q_2$};
\node[anchor=north west] at (2.4,4.7,0) {$Q_3$};
\end{tikzpicture} }
 &
 \small{\begin{tikzpicture}[scale=0.45]

 \draw[yslant=0.5,xslant=-1,pattern={crosshatch dots},pattern color=green](3,-1) rectangle +(1,1);
    \draw[yslant=0.5,pattern={crosshatch dots},pattern color=green] (4,-3) rectangle +(1,2);
     \draw[yslant=-0.5,pattern={crosshatch dots},pattern color=green](4,2) rectangle +(-3,-1);  
\draw[yslant=-0.5,pattern={crosshatch dots},pattern color=green](4,3) rectangle +(-1,-1); 

\draw[yslant=0.5,pattern={horizontal lines},pattern color=blue](4,-4) rectangle +(3,1);     
\draw[yslant=-0.5,pattern={horizontal lines},pattern color=blue](4,1) rectangle +(-3,-1); 
\draw[yslant=-0.5] (1,0) grid (4,3);
\draw[yslant=0.5] (4,-4) grid (7,-1);
\draw[yslant=0.5,xslant=-1] (3,-1) grid (6,2);
\node[anchor=south west] at (-0.25,1.5,0) {$S_1$};
\node[anchor=south west] at (-0.25,0.5,0) {$S_2$};
\node[anchor=south west] at (-0.25,-0.5,0) {$S_3$};
\node[anchor=north west] at (0.3,-0.6,0) {$T_3$};
\node[anchor=north west] at (1.3,-1.1,0) {$T_2$};
\node[anchor=north west] at (2.3,-1.6,0) {$T_1$};
\node[anchor=north west] at (0.4,3.7,0) {$Q_1$};
\node[anchor=north west] at (1.4,4.2,0) {$Q_2$};
\node[anchor=north west] at (2.4,4.7,0) {$Q_3$};
\end{tikzpicture}}\\
(a) & (b)\\
\end{tabular}\end{center}\caption{The red bundle corresponds to $T_1 \cap ((Q_1 \cap S_2) \cup S_1)$ and the blue one to $S_3$, as shown in (a). Agent $Q$ values the remaining items at $1/2-\epsilon$. In (b), the green bundle corresponds to $Q_1 \cap ((T_1 \cap S_1) \cup S_2)$ and the blue one to $S_1$. Agent $T$ values the remaining items at $1/2-\epsilon$.}
\label{fig:3agents322a}
\end{figure}
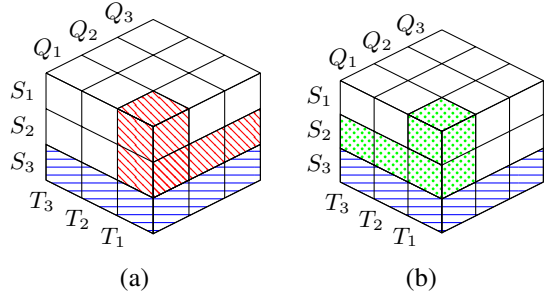
\end{center}

\begin{claim}
If $S\subset R_i$ is some $B^*$ set for $R$, then $R_i\setminus S$ is not a superset of some other $B^*$ set for $R$.
\end{claim}
\begin{proof}
First note that for any $B^*=R_i\cap((X_j\cap Y_{k})\cup Y_{k+1})$, it is $R_i\setminus B^* = R_i \cap (((X_{j+1} \cup X_{j+2})\cap Y_k)\cup Y_{k+2}) $. So, it is clear that $R_i\setminus B^*$ is a superset of bundles of the form  $R_i\cap((X_{j'}\cap Y_{k'})\cup Y_{k'+2})$, which are not considered as $B^*$. So, if $S$ is some $B^*$ set for $R$, considering $R_i$, $R_i\setminus S$ cannot be a superset of some $B^*$ set for $R$.
\end{proof}

Assume for the sake of contradiction that there exists a $(1,1/2,1/2)$-MMS allocation, $A=(A_S,A_T,A_Q)$. 
Due to the definition of the valuation functions we may assume that $A_S$ is one MMS bundle for $S$, and $A_T,A_Q$ are  subsets of one MMS bundle of $T$ and $Q$, respectively. Note that the values are symmetric and we make no assumptions for agent $S$. As a result we can prove the inaproximability of $(1/2,1,1/2)$-MMS and $(1/2,1/2,1)$-MMS.

Due to the symmetric instance, suppose w.l.o.g. that $A_S=S_3$ and $A_T\subseteq T_1$. Figure~\ref{fig:3agents322a}(a) shows an example where $T$ receives a minimum possible cardinality set;  Figure~\ref{fig:3agents322a}(b) shows the case where we considered $Q$ instead of $T$, to illustrate the symmetry. We will next show that it should be that $A_T\supset T_1\cap S_2$. 

If $|A_T|\geq 6$, clearly $A_T=T_1\setminus S_3 \supset T_1\cap S_2$.
If $|A_T| = 5$, based on the definition of the valuations, it should be that $T_1\setminus A_T$ is not a $B^*$ bundle for $T$. Note that $T_1\setminus A_T $ is a superset of $ T_1\cap S_3$ and has cardinality $4$. Therefore, $T_1\setminus A_T = T_1\cap ((Q_k\cap S_1)\cup S_3)$, for some $k\in \{1,2,3\}$, otherwise (i.e., if we replace $S_1$ with $S_2$) it would be a $B^*$ set for $T$. As a result, $A_T \supset T_1\cap S_2$. If $|A_T| = 4$, then $A_T$ should be a $B^*$ set for $T$. Since $A_T\cap S_3 =\emptyset$, $A_T$ should be of the form $A_T = T_1\cap ((Q_k\cap S_1)\cup S_2)$, for some $k\in \{1,2,3\}$, which means again that  $A_T \supset T_1\cap S_2$. 

So, in any case $A_S\cup A_T \supset (T_1\cap S_2)\cup S_3$. This leaves $Q$ with some $A_Q\subseteq Q_k$, for some $k\in\{1,2,3\}$, with $|A_Q|\leq 5$. Note that $Q_k\setminus A_Q = Q_k\cap ((T_1\cap S_2)\cup S_3)$, which $Q$ values by more than $1/2$, therefore $v_Q(A_Q)<1/2$, by the valuation definition, which contradicts the assumption of the existence of a $(1,1/2,1/2)$-MMS allocation.
\end{proof}

\subsection{Many Agents}
\label{sec:ManyAgentsImpossib}
In this section we show impossibility results regarding many agents with subadditive valuations. More specifically, we show the existence of $1/2+\epsilon$ inapproximability for any vector $\boldsymbol{d}$ (Lemma ~\ref{Lemma:upperhalf}), inapproximability up to any factor $\epsilon$ when agents have less than $n$ MMS bundles (Lemma~\ref{lem:nagentsImp1}) or for MMS($n,\ldots,n,\left\lfloor\frac{n}{3}\right\rfloor$) bundles (Lemma~\ref{lem:imposs(3,3,1)}). Our last Lemma shows that any inapproximability result for $1/2$-MMS$(\mathbf{d})$ implies an inapproximability for $(1/3+\epsilon)$-MMS$(\mathbf{d})$ (Lemma~\ref{Lemma:upperBound}).

\begin{lemma}\label{Lemma:upperhalf}
    For any number of agents $n$ and any vector $\boldsymbol{d}=(d_1,\ldots, d_n)$ there exists an instance, where agents have subadditive valuations,  where no $(1/2+\epsilon)$-MMS$(\boldsymbol{d})$ allocation exists, for any $\epsilon>0$. 
\end{lemma}
\begin{proof}
    We will construct a counterexample in which every set $S \subseteq M$ has value either $1/2$ or $1$ for each agent $i$. Those valuations are carefully constructed so that if an agent receives more than $1/2$ value, it is impossible for any other agent to receive more than $1/2$.

    We consider a set of $\prod_{i=1}^n d_i$ items, that we denote as $M = \left\{g(r_1,\ldots,r_n),r_i \le d_i\right\}$.
    Let the bundle $B_{i,j}=\left\{g(r_1,\ldots,r_n):r_i=j\right\}$, for all $i\leq  n$ and $j \le d_i$. We construct the functions of each agent $i$ and set of items $S \subseteq M$ as follows:
    $$v_i\left(S\right) = \begin{cases}
        0 & \text{if } S = \emptyset\\
        1 & \text{if } \exists j: B_{i,j} \subseteq S \\
        1/2 & \text{otherwise}\\
    \end{cases}$$
Observe that the functions are subbaditive, and that for each agent $i$, $v_i(M)=1$, and there exists a partition of the items into $d_i$ bundles, namely $B_{i,1},\ldots,B_{i,d_i}$, where agent $i$ has value $1$ for each of them.

Consider any arbitrary pair of agents $i \ne i'$, then  $B_{i,j} \cap B_{i',j'} = \left\{g(r_1,\ldots,r_n):r_i=j,r_{i'}=j'\right\}$ for every pair $(j,j')$. Assume for the sake of contradiction that there exists an $(1/2+\epsilon)$-MMS$(\boldsymbol{d})$ allocation, in which bundles $A_i$ and $A_{i'}$ are allocated to agents $i$ and $i'$, respectively. Since $v_i\left(A_i\right) > 1/2$, there exists some $B_{i,j} \subseteq A_i$ for some $j$, and similarly since $v_{i'}\left(A_{i'}\right) > 1/2$, there exists some $B_{i',j'} \subseteq A_{i'}$ for some $j'$. Hence, for each item $g = g(r_1,\ldots,r_n):r_i=j,r_{i'}=j'$ we have $g \in A_i$ and also $g \in A_{i'}$ but this contradicts that bundles $A_i$ and $A_{i'}$ belong to a valid allocation. Therefore, in this instance there is no $(1/2+\epsilon)$-MMS$(\boldsymbol{d})$ allocation.
\end{proof}

\begin{lemma} 
\label{lem:nagentsImp1}
For any number of agents $n$ and any vector $\boldsymbol{d}=(d_1,\ldots, d_n)$, with $d_i< n$ for all $i$, there exists an instance, where agents have subadditive valuations, where no $\epsilon$-MMS($\boldsymbol{d}$) allocation exists, for any $\epsilon>0$.
\end{lemma}

\begin{proof}
    This is straight forward. Assume the instance with $n$ agents, $n-1$ items and $v_i\left(S\right) = 1, \forall S \subseteq M,i \in N$. Each agent $i$ can partition $M$ into $d_i< n$ bundles that they value each with $1$. However, in every possible allocation there exists at least one agent without allocated items and her value is equal to $0$.
\end{proof}

\begin{corollary}
\label{cor:d_i<k}
    For any number of agents $n$, and vector $\mathbf{d}=(d_1,\ldots, d_n)$, if there exists a subset of agents $N'\subseteq N$ such that $d_i< |N'|$ for all $i\in N'$,
    then there exists an instance, where agents have subadditive valuations, where no $\epsilon$-MMS$(\mathbf{d})$ allocation exists, for any $\epsilon>0$. 
\end{corollary}

\begin{lemma}
\label{lem:imposs(3,3,1)}
    For any number of agents $n$, there exists an instance, where agents have subadditive valuations, where no $1/2$-MMS$(n,\ldots,n,\left\lfloor \frac{n}{3} \right\rfloor)$ allocation exists, or even no $\boldsymbol{\alpha}$-MMS$(n,\ldots,n,\left\lfloor \frac{n}{3} \right\rfloor)$ allocation exists, with $\alpha_n\geq 1/2$ and $\alpha_i>0$, for all $i<n$.
\end{lemma}

\begin{proof}
    Consider an instance with $n$ items, $M=\{g_1,\ldots,g_n\}$, where for any agents $i<n$, $v_i(S)=1$, for any $S\subseteq M$. Regarding agent $n$, we define the bundles $B_k=\{g_{3k-2}, g_{3k-1},g_{3k}\}$, for $1\leq k\leq \left\lfloor \frac{n}{3} \right\rfloor$, and then for any $S\subseteq M$, we define $v_n(S)$ as,
$$v_n(S)=\frac{\max_{1\leq k \leq \left\lfloor \frac{n}{3} \right\rfloor} \lvert S\cap B_k\rvert}3 \,.$$

Clearly, there are $\left\lfloor \frac{n}{3} \right\rfloor$ disjoint bundles of $M$, $(B_1,\ldots,B_{\left\lfloor \frac{n}{3}\right\rfloor})$, for which agent $n$ has value $1$.     
Consider any allocation $A$ that is $\boldsymbol{\alpha}$-MMS$(n,\ldots,n,\left\lfloor \frac{n}{3} \right\rfloor)$, for any $\boldsymbol{\alpha}$ with  $\alpha_n\geq 1/2$ and $\alpha_i>0$, for all $i<n$. Then, it should be that $|A_n|\geq 2$ which leaves at most $n-2$ items for the rest $n-1$ agents. So one agent $i<n$ would not receive any item in $A$, which violates the fact that $\alpha_i>0$, for all $i<n$.
\end{proof}

We show in the following lemma that all impossibility results of the non-existence of $1/2$-MMS allocations, extend to non-existence of $(1/3+\epsilon)$-MMS allocation, for any $\epsilon > 0$.

\begin{lemma} \label{Lemma:upperBound}
    Given a vector $\boldsymbol{d}=(d_1,\ldots, d_n)$, if there exists an instance $I$ of $n$ agents with subadditive valuations over $m$ items where no $1/2$-MMS$(\boldsymbol{d})$ allocation exists, then there is also an instance $I'$ of $n$ agents with subadditive valuations over $m$ items where no $(1/3+\epsilon)$-MMS$(\boldsymbol{d})$ allocation exists, for any $\epsilon > 0$. 
\end{lemma}

\begin{proof}
    We will show that we can transform any instance $I=(N,M,\boldsymbol{v})$ into $I'=(N,M,\boldsymbol{v}')$ such that for every $S \subseteq M$ and $i \in N$, it holds that $v_i(S) \ge 1/2 \Leftrightarrow v'_i(S) \ge 2/3$ and also  $v_i(S) < 1/2 \Leftrightarrow  v'_i(S) \le 1/3$. Then, if there was an $(1/3+\epsilon)$-MMS allocation in $I'$, this allocation would also be $2/3$-MMS in $I'$, and therefore  $1/2$-MMS in $I$. However, by the lemma's assumption, this would be a contradiction, and therefore, there is no $(1/3+\epsilon)$-MMS allocation in $I'$.

    We next show how to derive $\boldsymbol{v}'$ from $\boldsymbol{v}$. Given any subadditive valuation function $v_i$ of agent $i$, we construct the valuation function $v_i'$ as follows:
        $$v'_i\left(S\right) = \begin{cases}
        0 & \text{if } S = \emptyset\\
        1/3 & \text{if } 0 \le v_i(S) < 1/2\\
        2/3 & \text{if } 1/2 \le v_i(S) < 1\\
        1 & \text{if } v_i(S) \ge 1 \\
    \end{cases}$$
By definition, it hold that  $v_i(S) \ge 1/2 \Leftrightarrow v'_i(S) \ge 2/3$ and $v_i(S) < 1/2 \Leftrightarrow  v'_i(S) \le 1/3$. We will then show that $v'_i$ is a subadditive function. 

Consider any two sets of items, $S_1$ and $S_2$, and w.l.o.g. assume that $v'(S_1) \le v'(S_2)$. We will show that $v_i'(S_1)+v_i'(S_2) \geq v'_i(S_1 \cup S_2)$. This inequality holds if $S_1 = \emptyset$, since then $v_i'(S_1)+v_i'(S_2) = v_i'(S_2) = v'_i(S_1 \cup S_2)$. So, suppose that $S_1 \neq \emptyset$, which means that 
    $v_i'(S_1)\geq 1/3 $ and $ v_i'(S_2)\geq 1/3$. If $ v_i'(S_2)\geq 2/3$, then the subadditivity is not violated since $v_i'(S_1)+v_i'(S_2) \geq 1 \geq v'_i(S_1 \cup S_2)$, where the last inequality holds since by definition, $v_i'(S)\leq 1$, for any set $S$. So, at last suppose that $v_i'(S_1)=v_i'(S_2)= 1/3$. Then, $v_i(S_1)<1/2$ and $v_i(S_2)<1/2$, so by subbaditivity on $v_i$, $v_i(S_1 \cup S_2)<1$, which means that $v_i'(S_1 \cup S_2)\leq 2/3 = v_i'(S_1)+v_i'(S_2)$; so in this case the subadditivity is also preserved. Overall, $v_i'$ is subadditive and the lemma follows.  
\end{proof}

\begin{corollary}
    For any number of agents $n$, any number of goods $m$, and any vector $\boldsymbol{d}=(d_1,\ldots, d_n)$, if an $(1/3+\epsilon)$-MMS$(\boldsymbol{d})$ allocation, for some $\epsilon > 0$, is guaranteed to exist, for any subadditive valuation functions of the agents, then there is always an $1/2$-MMS$(\boldsymbol{d})$ allocation when agents have subadditive valuations.
\end{corollary}

 \section{Submodular Valuations}
\label{sec:Submod}

We present an improved upper bound for three agents with submodular valuation functions that rules out the existence of better than $2/3$-MMS allocations. The case of two agents has been previously studied in \cite{KKM23} and \cite{ChristodoulouChristoforidis}, establishing a tight answer of $2/3$ for the approximability of the maximin share.

\begin{theorem}
\label{thm:SubmodUB}
    There exists an instance of $3$ agents with submodular valuations and $6$ items for which there is no $(2/3+\epsilon)$-MMS allocation, for any $\epsilon>0$.
\end{theorem}

\begin{proof}
    Let the set of three agents be $N=\{a_1,a_2,a_3\}$ and let the set of $6$ items be $M = \{g_1, G_1, g_2, G_2, g_3, G_3\}$. Two of the agents, $a_1$ and $a_2$, share the same valuation function.
    
    There are two types of items: small items, denoted by $g_i$, and large items, $G_i$. For each agent $a\in N$ and bundle $S\subset M$, with $\lvert S \rvert = 1$ we define 
    $$v_a(S) =   \begin{cases}
    2/3, &\mbox{ if } S=\{G_i\} \\
      1/3, & \mbox{ if } S=\{g_i\} \\
      \end{cases}$$
   i.e. each agent values small items with $1/3$ and large items with $2/3$.

     Regarding the bundles of size two, each agent considers a bundle either as ``low'', for which they have value $2/3$, or as ``high'', for which they have value $1$. All agents consider the bundles of two small items as low, and of two large items as high. However, a bundle of a small and a large item may be considered differently by the agents. For each of the two agents that share the same valuation function, i.e., for $a \in \{a_1,a_2\}$, and for any bundle $S\subset M$, with $\lvert S \rvert = 2$, we define $$v_a(S) =   \begin{cases}
      1, & S=\{g_i,G_i\} \text{ or } S=\{G_i,G_j\} \quad \forall i,j \\
      2/3, & \text{otherwise} 
    \end{cases}$$
    For agent $a_3$, we shift the combinations of large and small items as follows: 
    
    $$  v_{a_3}(S) =   \begin{cases}
      1, & S=\{g_{(i\mod 3) +1},G_{i }\}  \text{ or } S=\{G_i,G_j\}\quad \forall i,j\\ 
      2/3, & \text{otherwise}
    \end{cases}$$
    
    The distinction between small and large bundles is illustrated in Figure \ref{fig:sizetwo} for all agents. Regarding the bundles with higher than $2$ cardinality, for any agent $a \in N$ and bundle $S\subseteq M,\lvert S\rvert \ge 3$ we define 
    $$v_a(S) =   \begin{cases}
    1, &\mbox{ if } |S|\geq 4 \mbox{ or } S=\{g_1,g_2,g_3\} \\
      \max_{S' \subseteq S,\lvert S' \rvert = 2}\left\{v_a\left(S\right)\right\}, & \mbox{ otherwise } 
      \end{cases}$$

    Observe that the MMS value of all agents is equal to $1$, because they can partition the items into three bundles, each of value $1$, by appropriately combining a small and a large item, i.e., $a_1$ and $a_2$ can form the bundles $\{g_1,G_1\},\{g_2,G_2\}$ and $\{g_3,G_3\}$, and $a_3$ can form the bundles $\{g_1,G_3\},\{g_2,G_1\}$ and $\{g_3,G_2\}$. By construction if a bundle has value $1$ for some agent then it  contains all small items, or two large items or it is one of her MMS bundles. %Note also that the MMS partition of agent $a_3$ is the MMS partition of the other agents shifted by $-1$.

    We show in the following claim that the valuation functions are submodular. 
    
    \begin{claim}
        The valuation functions defined above are submodular functions.
    \end{claim}
    \begin{proof}
    
    We first argue that the function is monotone. Indeed, the value for any set of cardinality $1$ is at most $2/3$, of cardinality $2$ is between $2/3$ and $1$, of cardinality $3$ is at most $1$ and cannot be less than the value of any of its subsets by definition, and for greater cardinality is $1$. 
    
    We proceed by showing that the valuation $v_{a_1}=v$ satisfies the submodularity property, i.e., for any two sets $S$ and $T$, such that $S \subset T \subset M$, and for any item $g\in M\setminus T$ it holds \begin{equation}
    v\left(S \cup \{g\}\right) - v(S) \ge v(T \cup \{g\}) - v(T)\,. 
    \label{eq:sub}
    \end{equation} 
    For $v_{a_3}$ the arguments are the same by shifting the indices of the small items. 
    %  from definition for every $\lvert S'\rvert \ge 3$ the function is monotone. If $\lvert S\rvert = 1$ and $\lvert S'\rvert = 2$ then $v_a(S) \le 2/3 \le v_a(S')$. Hence there must hold $v_a(S' \cup \{g\}) - v_a(S')>0$ and thus $S$ cannot contain more than $3$ items. As a result $S$ is empty, contains $1$ or $2$ items. 

Note that if $g$ is a small item, its marginal contribution to any set of items is either $0$, or $1/3$. If $g$ is a large item, its marginal contribution to any set of items may be $0$, or $1/3$, or $2/3$. Further, note that if $|T|\geq 4$, all items has a zero marginal contribution, and hence inequality \eqref{eq:sub} holds. So, we assume $|T|\leq 3$, and we distinguish between the different cases for the cardinality of $S$, which should be at most two. 

    {\bf Case 1: $|S|=0$.} In this case $v(S\cup \{g\})-v(S)=v(\{g\})$ which is the maximum possible contribution of any item,  and so, inequality \ref{eq:sub} holds.

    {\bf Case 2: $|S|=1$.} In this case $T$ contains at least two items and the marginal contribution of $g$ to $T$ is at most $1/3$, i.e., $v(T \cup \{g\}) - v(T) \le 1/3$. If $S$ contains a small item, then $v(S\cup\{g\})-v(S)\geq 2/3-1/3=1/3$, since $S\cup\{g\}$ contains two items and the value for those sets is at least $2/3$; so inequality \eqref{eq:sub} holds. If $S=\{G_i\}$, for some $i\in\{1,2,3\}$, then the only case that $v(S\cup\{g\})-v(S)=0$, is for $g=g_j$, for some $j\neq i$. 
    
    We will show that in that case $v(T \cup \{g\}) - v(T)=0$. Suppose on the contrary that $v(T \cup \{g\})=1$ and $v(T)=2/3$. Since $G_i\in T$, $T$ cannot contain any other large item or $g_i$. Moreover, $g\notin T$, where recall that $g=g_j$, for some $j\neq i$. Therefore, $T=\{G_i,g_k\}$, for the $k\notin\{i,j\}$. However, by definition of the valuations, $v(T \cup \{g\}) = v(\{G_i,g_j,g_k\})=2/3$, for $i,j,k$ being different pairwise. So,  $v(T \cup \{g\}) - v(T)=0$, meaning that inequality \eqref{eq:sub} holds in that case as well.   

    {\bf Case 2: $|S|=2$.} In this case, $v(T \cup \{g\})=1$, since $T \cup \{g\}$ has cardinality four, and also $v(S)\geq 2/3$, since $S$ has cardinality two. So, the only way that \eqref{eq:sub} is violated is if $v(T)=v(S)=v(S\cup\{g\}=2/3$. The only sets of cardinality three with value $2/3$ are of the form $\{G_i,g_j,g_k\}$, for $i,j,k$ being different pairwise. So, suppose that $T=\{G_i,g_j,g_k\}$ and $S$ is either $\{g_j,g_k\}$, or contains $G_i$ If $g=g_i$, then no matter what $S$ is, $v(S\cup\{g\})=1$. Similarly, if $g=G_j$, for any $j\neq i$, $v(S\cup\{g\})=1$. In any case, our assumption is violated and therefore inequality \eqref{eq:sub} always holds in that case as well.  
    \end{proof}

    Next we show that in this instance there is no $(2/3+\epsilon)$-MMS allocation, for any $\epsilon>0$. Assume for the sake of contradiction that there exists an $(2/3+\epsilon)$-MMS allocation. In this allocation, no agent may receive a single item, since each item has value at most $2/3$ for each agent. For the same reason, no agent may receive three or more items, since then some other agent must recieve at most one item that she values by at most $2/3$. As a result each agent gets exactly two items. If in the allocation some agent recieves two large items then some other agent receives only small items and her value is at most $2/3$. Thus, each agent gets a small and a large item. If those items do not correspond to one of their MMS bundles, they value their allocated set by at most $2/3$. Hence, each agent must receive one of their MMS bundles. Due to the construction we can allocate to at most two agents one of their MMS bundles, meaning that there is no $(2/3+\epsilon)$-MMS allocation in this instance.
\end{proof}

% \begin{figure}
% \centering
%     \begin{tabular}{c|c c c c c c}
%      & $g_1$ & $G_1$ & $g_2$ & $G_2$ & $g_3$ & $G_3$\\
%      \hline
%      $a_1$ & 1/3 & 2/3 & 1/3 & 2/3 & 1/3 & 2/3 \\
%      $a_2$ & 1/3 & 2/3 & 1/3 & 2/3 & 1/3 & 2/3 \\
%      $a_3$ & 1/3 & 2/3 & 1/3 & 2/3 & 1/3 & 2/3 \\
% \end{tabular}
%     \caption{Values of singletons. Items $g_1,g_2$ and $g_3$ are the small items and items $G_1,G_2$ and $G_3$ are the large.}
%     \label{fig:singletons}
% \end{figure}

\begin{figure}
\centering
\begin{tabular}{c|c c c c c c c c c}
     & $g_1G_1$ & $g_1G_2$ & $g_1G_3$ & $g_2G_2$ & $g_2G_3$ & $g_2G_1$ & $g_3G_3$ & $g_3G_1$ & $g_3G_2$ \\
     \hline
     $a_1$ & 1 &  2/3 &  2/3 & 1 &  2/3 &  2/3 & 1 &  2/3 &  2/3\\
     $a_2$ & 1 &  2/3 &  2/3 & 1 &  2/3 &  2/3 & 1 &  2/3 &  2/3 \\
     $a_3$ & 2/3 & 2/3 & 1 & 2/3 & 2/3 & 1 & 2/3 & 2/3 & 1 
     
\end{tabular}

\begin{tabular}{c|c c}
    & $g_ig_j$ & $G_iG_j$ \\
    \hline
    $a_1$ & 2/3 & 1 \\
    $a_2$ & 2/3 & 1 \\
    $a_3$ & 2/3 & 1
\end{tabular}
    \caption{If a bundle is large for some agent then either it contains two large items or it is one of her MMS bundles.}
    \label{fig:sizetwo}
\end{figure}

 \section{Conclusion}
We study the existence of approximate MMS allocations in the case of subadditive valuations and few agents. We showed the existence of $1/2$-MMS allocations for at most four agents with subadditive valuations, as well as for multiple agents when they have one of two types of valuations.
The most challenging question of whether a constant-factor MMS approximations always exist for $n > 4$ agents, remains open. One hopes that our insights and technical lemmas might help pave the way for more general positive results. The study of $\boldsymbol{\alpha}$-MMS(${\bf d}$) and $\boldsymbol{\alpha}$-MMS(${\bf P}$) was proved useful for providing approximate MMS guarantees, but we believe they are of independent interest and deserve further investigation.

 \section*{Acknowledgements}

This work has been partially supported by project MIS 5154714 of the National Recovery and Resilience Plan Greece 2.0 funded
by the European Union under the NextGenerationEU Program.
  
 \bibliographystyle{plain}
 \bibliography{mms}

\begin{thebibliography}{10}

\bibitem{Aigner-HorevSH22}
Elad Aigner{-}Horev and Erel Segal{-}Halevi.
\newblock Envy-free matchings in bipartite graphs and their applications to fair division.
\newblock {\em Inf. Sci.}, 587:164--187, 2022.

\bibitem{AkramiGarg24}
Hannaneh Akrami and Jugal Garg.
\newblock Breaking the 3/4 barrier for approximate maximin share.
\newblock In {\em {SODA}}, pages 74--91. {SIAM}, 2024.

\bibitem{AkramiGST23}
Hannaneh Akrami, Jugal Garg, Eklavya Sharma, and Setareh Taki.
\newblock Simplification and improvement of {MMS} approximation.
\newblock In {\em {IJCAI}}, pages 2485--2493. ijcai.org, 2023.

\bibitem{AkramiGST24}
Hannaneh Akrami, Jugal Garg, Eklavya Sharma, and Setareh Taki.
\newblock Improving approximation guarantees for maximin share.
\newblock In {\em {EC}}, page 198. {ACM}, 2024.

\bibitem{AkramiMehlhornSeddighinShahkarami23}
Hannaneh Akrami, Kurt Mehlhorn, Masoud Seddighin, and Golnoosh Shahkarami.
\newblock Randomized and deterministic maximin-share approximations for fractionally subadditive valuations.
\newblock In {\em NeurIPS}, 2023.

\bibitem{AmanatidisABFLMVW23Survey}
Georgios Amanatidis, Haris Aziz, Georgios Birmpas, Aris Filos{-}Ratsikas, Bo~Li, Herv{\'{e}} Moulin, Alexandros~A. Voudouris, and Xiaowei Wu.
\newblock Fair division of indivisible goods: Recent progress and open questions.
\newblock {\em Artif. Intell.}, 322:103965, 2023.

\bibitem{AmanatidisMNS17}
Georgios Amanatidis, Evangelos Markakis, Afshin Nikzad, and Amin Saberi.
\newblock Approximation algorithms for computing maximin share allocations.
\newblock {\em {ACM} Trans. Algorithms}, 13(4):52:1--52:28, 2017.

\bibitem{BabaioffNT21}
Moshe Babaioff, Noam Nisan, and Inbal Talgam{-}Cohen.
\newblock Competitive equilibrium with indivisible goods and generic budgets.
\newblock {\em Math. Oper. Res.}, 46(1):382--403, 2021.

\bibitem{BarmanKrishnamurthy20}
Siddharth Barman and Sanath~Kumar Krishnamurthy.
\newblock Approximation algorithms for maximin fair division.
\newblock {\em {ACM} Trans. Economics and Comput.}, 8(1):5:1--5:28, 2020.

\bibitem{BudishMMS}
Eric Budish.
\newblock The combinatorial assignment problem: Approximate competitive equilibrium from equal incomes.
\newblock {\em Journal of Political Economy}, 119(6):1061--1103, 2011.

\bibitem{ChekuriKKM24}
Chandra Chekuri, Pooja Kulkarni, Rucha Kulkarni, and Ruta Mehta.
\newblock 1/2-approximate {MMS} allocation for separable piecewise linear concave valuations.
\newblock In {\em {AAAI}}, pages 9590--9597. {AAAI} Press, 2024.

\bibitem{ChristodoulouChristoforidis}
George Christodoulou and Vasilis Christoforidis.
\newblock Fair and truthful allocations under leveled valuations.
\newblock {\em CoRR}, abs/2407.05891, 2024.

\bibitem{FeigeNorkinMMS}
Uriel Feige and Alexey Norkin.
\newblock Improved maximin fair allocation of indivisible items to three agents.
\newblock {\em CoRR}, abs/2205.05363, 2022.

\bibitem{GargTaki21}
Jugal Garg and Setareh Taki.
\newblock An improved approximation algorithm for maximin shares.
\newblock {\em Artif. Intell.}, 300:103547, 2021.

\bibitem{GhodsiHSSY21}
Mohammad Ghodsi, Mohammad~Taghi Hajiaghayi, Masoud Seddighin, Saeed Seddighin, and Hadi Yami.
\newblock Fair allocation of indivisible goods: Improvement.
\newblock {\em Math. Oper. Res.}, 46(3):1038--1053, 2021.

\bibitem{GhodsiHSSY22}
Mohammad Ghodsi, Mohammad~Taghi Hajiaghayi, Masoud Seddighin, Saeed Seddighin, and Hadi Yami.
\newblock Fair allocation of indivisible goods: Beyond additive valuations.
\newblock {\em Artif. Intell.}, 303:103633, 2022.

\bibitem{GourvesMonnot19}
Laurent Gourv{\`{e}}s and J{\'{e}}r{\^{o}}me Monnot.
\newblock On maximin share allocations in matroids.
\newblock {\em Theor. Comput. Sci.}, 754:50--64, 2019.

\bibitem{HosseiniS21}
Hadi Hosseini and Andrew Searns.
\newblock Guaranteeing maximin shares: Some agents left behind.
\newblock In {\em {IJCAI}}, pages 238--244. ijcai.org, 2021.

\bibitem{HosseiniSSH22}
Hadi Hosseini, Andrew Searns, and Erel Segal{-}Halevi.
\newblock Ordinal maximin share approximation for goods.
\newblock {\em J. Artif. Intell. Res.}, 74, 2022.

\bibitem{Hummel:HSS24}
Halvard Hummel.
\newblock Maximin shares in hereditary set systems.
\newblock {\em CoRR}, abs/2404.11582, 2024.

\bibitem{KKM23}
Pooja Kulkarni, Rucha Kulkarni, and Ruta Mehta.
\newblock Maximin share allocations for assignment valuations.
\newblock In {\em {AAMAS}}, pages 2875--2876. {ACM}, 2023.

\bibitem{KurokawaProcacciaWang18}
David Kurokawa, Ariel~D. Procaccia, and Junxing Wang.
\newblock Fair enough: Guaranteeing approximate maximin shares.
\newblock {\em J. {ACM}}, 65(2):8:1--8:27, 2018.

\bibitem{LiVetta21}
Zhentao Li and Adrian Vetta.
\newblock The fair division of hereditary set systems.
\newblock {\em {ACM} Trans. Economics and Comput.}, 9(2):12:1--12:19, 2021.

\bibitem{SeddighinSeddighin24}
Masoud Seddighin and Saeed Seddighin.
\newblock Improved maximin guarantees for subadditive and fractionally subadditive fair allocation problem.
\newblock {\em Artif. Intell.}, 327:104049, 2024.

\bibitem{UziahuFeigeMMS}
Gilad~Ben Uziahu and Uriel Feige.
\newblock On fair allocation of indivisible goods to submodular agents.
\newblock {\em CoRR}, abs/2303.12444, 2023.

\end{thebibliography}

\end{document}